%% file: optimal_nm.tex
\documentclass[11pt]{article}
\usepackage[margin=1in]{geometry}

\usepackage{amsmath,amsthm,amssymb,verbatim,algorithm2e,mathtools}
\usepackage[noadjust]{cite}
\usepackage[usenames,dvipsnames]{xcolor}
\usepackage{thmtools}
\usepackage{thm-restate}
\usepackage{array}

\usepackage{algpseudocode}

\usepackage[colorlinks=true, linkcolor=red,citecolor=ForestGreen,urlcolor = black]{hyperref}

\newcommand{\remove}[1]{}
\newtheorem{thm}{Theorem}[section]
\newtheorem{claim}[thm]{Claim}
\newtheorem{lemma}[thm]{Lemma}
\newtheorem{define}[thm]{Definition}
\newtheorem{cor}[thm]{Corollary}

\newtheorem{remark}[thm]{Remark}

\newtheorem{THM}{Theorem}

\renewcommand{\remove}[1]{}
\newcommand{\polylog}{{\rm polylog}}
\newcommand{\poly}{{\rm poly}}

\newcommand{\eps}{{\varepsilon}}
\renewcommand{\l}{\left}
\renewcommand{\r}{\right}
\newcommand{\ep}{{\epsilon}}

\newcommand{\la}{{\lambda}}

\newcommand{\comments}[1]{}

\newcommand{\nmExt}{\textnormal{nmExt}}

\newcommand{\D}{\mathbf{\mathcal{D}}}

\newcommand{\scirc}{\hspace{0.1cm}\circ \hspace{0.1cm}}

\newcommand{\slice}{\textnormal{Slice}}

\newcommand{\samp}{\textnormal{Samp}}

\newcommand{\Ext}{\textnormal{Ext}}

\newcommand{\laExt}{\textnormal{laExt}}

\newcommand{\support}{\textnormal{support}}

\newcommand{\reduce}{\textnormal{reduce}}

\newcommand{\nipm}{\textnormal{NIPM}}
\newcommand{\ipm}{\textnormal{IPM}}

\newcommand{\E} {\mathbf{E}}

\newcommand{\N}{\mathcal{N}}
\newcommand{\A}{\mathcal{A}}

\newcommand{\U}{\mathbf{U}}

\newcommand{\X}{\mathbf{X}}
\newcommand{\Y}{\mathbf{Y}}
\newcommand{\xb}{\mathbf{x}}

\newcommand{\zo}{\{0, 1\}}
\newcommand{\pr}{\mathbf{Pr}}

\newcommand{\s}{\mathbf{S}}
\newcommand{\W}{\mathbf{W}}
\newcommand{\V}{\mathbf{V}}
\newcommand{\rr}{\mathbf{R}}

\newcommand{\Z}{\mathbf{Z}}
\newcommand{\Q}{\mathbf{Q}}
\newcommand{\w}{\mathbf{W}}

\newcommand{\flip}{\textnormal{flip-flop}}
\newcommand{\adv}{\textnormal{advGen}}

\makeatletter
\def\old@comma{,}
\catcode`\,=13
\def,{%
  \ifmmode%
    \old@comma\discretionary{}{}{}%
  \else%
    \old@comma%
  \fi%
}
\newcommand{\x}[1]{{}$\kern-2\mathsurround${}
  \binoppenalty10000 \relpenalty10000 #1{}$\kern-2\mathsurround${}}
  
\makeatother
\def\draft{1}   

\ifnum\draft=1 
    \def\ShowAuthNotes{1}
\else
    \def\ShowAuthNotes{0}
\fi

\ifnum\ShowAuthNotes=1
\newcommand{\authnote}[2]{{ \footnotesize \bf{\color{red}[#1's Note: {\color{blue}#2}]}}}
\else
\newcommand{\authnote}[2]{}
\fi

\makeatletter
\def\sand{%
  \end{tabular}%
  \hskip 0.5em \@plus.17fil\relax
  \begin{tabular}[t]{c}}
\makeatother

\begin{document} 
\title{Explicit Non-Malleable Extractors, Multi-Source Extractors  and Almost Optimal Privacy Amplification Protocols}
\author{  Eshan Chattopadhyay\thanks{
Partially supported by NSF Grant CCF-1526952 and a Dissertation Writing Fellowship awarded by UT Austin.
}\\ Department of Computer Science,\\ University of Texas at Austin \\ \href{mailto:eshanc@cs.utexas.edu}{eshanc@cs.utexas.edu}
 \and Xin Li\\ Department of Computer Science\\John Hopkins University \\  \href{mailto:lixints@cs.jhu.edu}{lixints@cs.jhu.edu}}

 \maketitle
\thispagestyle{empty}

\input{abstract}
\clearpage 
\setcounter{page}{1}
\input{intro}

\input{preliminaries}

\input{nm_extractor}

\bibliographystyle{alpha}
\bibliography{optimal_nm}
\appendix 

\end{document}

%% file: abstract.tex
\begin{abstract}
We make progress in the following three problems: 1. Constructing optimal seeded non-malleable extractors; 2. Constructing optimal privacy amplification protocols with an active adversary, for any possible security parameter; 3. Constructing extractors for independent weak random sources, when the min-entropy is extremely small (i.e., near logarithmic). 

For the first two problems, the best known non-malleable extractors by Chattopadhyay, Goyal and Li \cite{CGL15}, and by Cohen \cite{Coh15nm,Coh16a} all require seed length and min-entropy at least $\log^2 (1/\eps)$, where $\eps$ is the error of the extractor. As a result, the best known explicit privacy amplification protocols with an active adversary, which achieve 2 rounds of communication and optimal entropy loss in \cite{Li15b,CGL15}, can only handle security parameter up to $s=\Omega(\sqrt{k})$, where $k$ is the min-entropy of the shared secret weak random source. For larger $s$ the best known protocol with optimal entropy loss in \cite{Li15b} requires $O(s/\sqrt{k})$ rounds of communication.

In this paper we give an explicit non-malleable extractor that only requires seed length and min-entropy $\log^{1+o(1)} (n/\eps)$, which also yields a 2-round privacy amplification protocol with optimal entropy loss for security parameter up to $s=k^{1-\alpha}$ for any constant $\alpha>0$.

For the third problem, previously the best known extractor which supports the smallest min-entropy due to Li \cite{Li13b}, requires min-entropy $\log^{2+\delta} n$ and uses $O(1/\delta)$ sources, for any constant $\delta>0$. A very recent result by Cohen and Schulman \cite{Coh16b} improves this, and constructed explicit extractors that use $O(1/\delta)$ sources for min-entropy $\log^{1+\delta} n$, any constant $\delta>0$. In this paper we further improve their result, and give an explicit extractor that uses $O(1)$ (an absolute constant) sources for min-entropy $\log^{1+o(1)} n$.

The key ingredient in all our constructions is a generalized, and much more efficient version of the independence preserving merger introduced in \cite{Coh16b}, which we call \emph{non-malleable independence preserving merger}. Our construction of the merger also simplifies that of \cite{Coh16b}, and may be of independent interest.
\end{abstract}

%% file: intro.tex
\section{Introduction}
The theory of \emph{randomness extractors} is a broad area and a fundamental branch of the more general study of pseudorandomness. Informally, randomness extractors are functions that transform biased probability distributions (weak random sources) into almost uniform probability distributions. Here we measure the entropy of a weak random source by the standard min-entropy. A source $\X$  is said to have min-entropy $k$ if for any $x$, $\Pr[\X=x] \le 2^{-k}$. An $(n,k)$-source $\X$ is a distribution on $n$ bits with min-entropy at least $k$.

It is well known that it is impossible to construct deterministic randomness extractors when the input is just one (arbitrary) weak random source, even if the min-entropy is as large as $n-1$. A natural relaxation is then to give the extractor a short independent uniform seed, and such extractors are called \emph{seeded extractors}. With this relaxation it is indeed possible to construct extractors that work for any weak random source with essentially any min-entropy. We now formally define such extractors.

\begin{define}
 The statistical distance between two distributions $\D_1$ and $\D_2$ over some universal set $\Omega$ is  defined as $|\D_1-\D_2| = \frac{1}{2}\sum_{d \in \Omega}|\pr[\D_1=d]- \pr[\D_2=d]|$. We say $\D_1$ is $\epsilon$-close to $\D_2$ if $|\D_1-\D_2| \le \epsilon$ and denote it by $\D_1 \approx_{\epsilon} \D_2$.
 \end{define}

\begin{define}[\cite{NZ96}]
 A function $\Ext:\{0,1\}^{n} \times \{ 0,1\}^d \rightarrow \{ 0,1\}^m$ is a seeded extractor for min-entropy $k$ and error $\epsilon$ if for any source $\X$ of min-entropy $k$, $|\Ext(\X,\U_d) - \U_m| \le \epsilon$. $\Ext$ is strong if in addition $|(\Ext(\X,\U_d), \U_d) - (\U_m,\U_d) | \le \epsilon$, where $\U_m$ and $\U_d$ are independent. 
 \end{define}
Through a long line of research we now have explicit constructions of seeded extractors with almost optimal parameters \cite{LRVW03,GUV09,DKSS09}.

In recent years, there has been much interest in the study of two other kinds of randomness extractors. The first one, known as \emph{non-malleable extractors} (introduced by Dodis and Wichs  \cite{DW09}), is a generalization of strong seeded extractors. 

\begin{define}[Non-malleable extractor] A function $\nmExt:\{0,1\}^n \times \{ 0,1\}^d \rightarrow \{ 0,1\}^m$ is a  $(k,\epsilon)$-non-malleable extractor if the following holds: For any $(n,k)$-source $\X$, an independent uniform seed $\Y$   on $d$ bits and any function $\A : \{0,1\}^d \rightarrow \{0,1\}^d $ with no fixed points,\footnote{i.e., for any $x$, $\A(x) \neq x$} $$   |(\nmExt(\X,\Y),\nmExt(X,\A(\Y)),\Y)- (\U_m,\nmExt(\X,\A(\Y)),\Y) | \le \epsilon.$$
\end{define}

The second one, known as \emph{multi-source extractors} (first studied by Chor and Goldreich \cite{CG88}), is another natural relaxation of deterministic extractors for one weak random source, in the sense that now the input to the extractor are several (at least two) independent weak random sources. Curiously, although this problem was first studied around 30 years ago, it was not until recently that significant progress has been achieved.

The above two kinds of extractors are closely related, and in many cases techniques used for one can also be used to improve the constructions of the other. These connections have been demonstrated in a number of works (e.g., \cite{Li12b,Li13a,Li13b,Li15c,CZ15}).

We now briefly discuss the motivations for these two kinds of extractors.

\subsection{Non-malleable extractors and privacy amplification}
The initial motivation for non-malleable extractors comes from the problem of privacy amplification with an active adversary \cite{BBR88,Mau92,BBCM95}. As a basic problem in information theoretic cryptography, privacy amplification deals with the case where two parties want to communicate with each other to convert their shared secret weak random source $\X$ into shared secret nearly uniform random bits. On the other hand, the communication channel is watched by an adversary Eve, who has unlimited computational power. To make this task possible, we assume two parties have local (non-shared) uniform random bits. 

If Eve is passive (i.e., can only see the messages but cannot change them), this problem can be solved easily by applying the aforementioned strong seeded extractors. However, in the case where Eve is active (i.e., can arbitrarily change, delete and reorder messages), the problem becomes much more complicated. The major challenge here is to design a protocol that uses as few number of interactions as possible, and outputs a uniform random string $\rr$ that has length as close to $H_{\infty}(\X)$ as possible (the difference is called \emph{entropy loss}). A bit more formally, we pick a security parameter $s$, and if the adversary Eve remains passive during the protocol then the two parties should achieve shared secret random bits that are $2^{-s}$-close to uniform. On the other hand, if Eve is active, then the probability that Eve can successfully make the two parties output two different strings without being detected should be at most $2^{-s}$. We refer the readers to \cite{DLWZ11} for a formal definition.

There has been a long line of work on this problem \cite{MW07,dkrs,DW09,RW03,KR09,ckor,DLWZ11,CRS12,Li12a,Li12b,Li15b,ADD14}. When the entropy rate of $\X$ is large, i.e., bigger than $1/2$, there are known protocols that take only one round (e.g., \cite{MW07,dkrs}). However these protocols all have very large entropy loss. When the entropy rate of $\X$ is smaller than $1/2$, \cite{DW09} showed that no one round protocol exists; furthermore the length of $\rr$ has to be at least $O(s)$ smaller than $H_{\infty}(\X)$. Thus, the natural goal is to design a two-round protocol with such optimal entropy loss. However, all protocols before the work of \cite{DLWZ11} either need to use $O(s)$ rounds, or need to incur an entropy loss of $O(s^2)$. 

In \cite{DW09}, Dodis and Wichs showed that explicit constructions of the aforementioned non-malleable extractors can be used to give two-round privacy amplification protocols with optimal entropy loss. Using the probabilistic method, they also showed that non-malleable extractors exist when $k>2m+2\log(1/\eps) + \log d + 6$ and $d>\log(n-k+1) + 2\log (1/\eps) + 5$. However, they were not able to give explicit constructions even for min-entropy $k=n-1$.\ The first explicit construction of non-malleable extractors appeared in \cite{DLWZ11}, with subsequent improvements in \cite{CRS12,Li12a,DY12,Li12b,ADD14}. All these constructions require the min-entropy  of the weak source to be bigger than $0.49n$, and thus only give two-round privacy amplification protocols with optimal entropy loss for such min-entropy. Together with some other ideas, \cite{DLWZ11} also gives $\poly(1/\delta)$ round protocols with optimal entropy loss for min-entropy $k \geq \delta n$, any constant $\delta>0$. This was subsequently improved by one of the authors in \cite{Li12b} to obtain a two-round protocol with optimal entropy loss for min-entropy $k \geq \delta n$, any constant $\delta>0$. In the general case, using a relaxation of non-malleable extractors called non-malleable condensers, one of the authors \cite{Li15b} also obtained a two-round protocol with optimal entropy loss for min-entropy $k \geq C\log^2 n$, some constant $C>1$, as long as the security parameter $s$ satisfies $k \geq Cs^2$. For larger security parameter, the best known protocol with optimal entropy loss in \cite{Li12b} still takes $O(s/\sqrt{k})$ rounds.

In a recent work, Chattopadhyay, Goyal and Li \cite{CGL15} constructed explicit non-malleable extractors with error $\eps$, for min-entropy $k =\Omega( \log^2{(n/\epsilon)})$ and seed-length $d=O(\log^2(n/\epsilon))$. This gives an alternative protocol matching that of \cite{Li12b}. Subsequently, Cohen \cite{Coh15nm} improved this result, and constructed non-malleable extractors with seed length $d= O(\log(n/\epsilon)\log((\log n)/\epsilon))$ and min-entropy $k = \Omega(\log(n/\epsilon)\log((\log n)/\epsilon))$. In this work, he also gave another construction that worked for $k=n/(\log n)^{O(1)}$ with seed-length $O(\log n)$. In a follow up, Cohen \cite{Coh16a}  constructed non-malleable extractors with seed length $d=O(\log n+ \log^3(1/\epsilon))$ and min-entropy $k=\Omega(d)$. However, in terms of the general error parameter $\eps$, all of these results require min-entropy and seed length at least $\log^2 (1/\eps)$, thus none of them can be used to improve the privacy amplification protocols in \cite{Li15b}.

A recent work by Aggarwal, Hosseini and Lovett \cite{AHL15} obtained some conditional results. In particular, they used a weaker variant of non-malleable extractors to construct privacy amplification protocols with optimal entropy loss for $k=\Omega(\log(1/\epsilon) \log n)$ assuming a conjecture in additive combinatorics.

\subsection{Multi-source extractors for independent sources}
As mentioned before, Chor and Goldreich \cite{CG88} introduced the problem of designing extractors for two or more independent sources. Explicit constructions of such extractors can also be used in explicit constructions of Ramsey graphs (\cite{BRSW12,Coh15b,CZ15}). A simple probabilistic argument shows the existence of two-source extractors for min-entropy $k \ge  \log n +O(1)$. However, explicit constructions of such functions are extremely challenging.
  
Chor and Goldreich \cite{CG88}  proved that the inner-product function is a two-source extractor for min-entropy greather than $n/2$.  It was not until $20$ years later when  Bourgain \cite{B2} broke the entropy rate $1/2$ barrier and  constructed a two-source extractor for min-entropy $0.49 n$.  Raz \cite{Raz05} obtained another construction which requires one source with min-entropy more than $n/2$ and the other source with min-entropy $O(\log n)$. Recently, Chattopadhyay and Zuckerman \cite{CZ15} improved the situation substantially by constructing two-source extractors for min-entropy $k \geq \polylog(n)$, with subsequent improvements obtained by Li \cite{Li:2source} and Meka \cite{Mek:resil}. The ultimate goal here is to obtain two-source extractors matching the entropy bound given by the probabilistic method. 

If we allow the extractor to have a constant number of sources instead of just two sources, then an exciting line of work \cite{BIW,BKSSW10,Rao06,BRSW12,RZ08,Li11,Li13a,Li13b,Li15c,Coh15}  constructed extractors with excellent parameters. However, the smallest entropy these constructions can achieve is $\log^{2+\delta}n$ for any constant $\delta>0$ \cite{Li13b}, which uses $O(1/\delta)+O(1)$ sources. In a very recent work, Cohen and Schulman \cite{Coh16b} managed to break this ``quadratic" barrier, and constructed extractors for $O(1/\delta)+O(1)$ sources, each with min-entropy $\log^{1+\delta}n$. 

\subsection{Our results}
\noindent \textbf{Non-Malleable Extractors}\hspace{0.2cm}Our first result is a new construction of non-malleable extractors that breaks the $\log^2 (1/\eps)$ barrier for min-entropy and seed length. Specifically, we have the following theorem.

\begin{THM} \label{thm:nmext1} There exists a constant $C>0$ s.t for all $n, k \in \mathbb{N}$ and any  $\epsilon>0$, with $k \ge \log(n/\epsilon)2^{C\sqrt{\log \log(n/\epsilon)}}$, there exists an explicit $(k,\epsilon)$-non-malleable extractor $\nmExt:\zo^n \times \zo^d \rightarrow \zo^{m}$, where $d= \log(n/\epsilon)2^{C\sqrt{\log \log(n/\epsilon)}}$ and $m= k/2^{\sqrt{\log \log (n/\epsilon)}}$.
\end{THM}

We also construct a non-malleable extractor with seed-length $O(\log n)$ for  min-entropy $k=\Omega(\log n)$ and $\epsilon \ge 2^{-\log^{1-\beta}(n)}$ for any $\beta>0$. Prior to this, explicit non-malleable extractors with seed-length $O(\log n)$ either requires min-entropy at least $n/\poly(\log n)$ \cite{Coh15nm} or requires $\epsilon \ge 2^{-\log^{1/3}(n)}$ \cite{Coh16a}.
\begin{THM}There exists a constant $C>0$ s.t for  and all $n, k \in \mathbb{N}$ with $k \ge C\log n$, any constant $0<\beta<1$, and any  $\epsilon \ge 2^{-\log^{1-\beta}(n)}$, there exists an explicit $(k,\epsilon)$-non-malleable extractor $\nmExt:\zo^n \times \zo^d \rightarrow \zo^{m}$, where $d= O(\log n)$ and $m=\Omega(\log(1/\epsilon))$.
\end{THM}

\begin{remark}
A careful examination reveals that our seed length and min-entropy requirement are better than those of \cite{Coh15nm,Coh16a} in all cases except the case that $\eps$ is large enough (e.g., $\epsilon \ge 2^{-\log^{1/3}(n)}$), where both \cite{Coh16a} and our results require seed length and min-entropy $O(\log n)$.
\end{remark}

Note that given any error parameter $\eps$, our non-malleable extractor in Theorem~\ref{thm:nmext1} only requires min-entropy and seed length $\log^{1+o(1)} (n/\eps)$.  
 
 
 We also show how to further lower the min-entropy requirement of the non-malleable extractor in Theorem $\ref{thm:nmext1}$ at the expense of using a larger seed. We complement this result by constructing another non-malleable extractor with shorter seed-length than in Theorem $\ref{thm:nmext1}$ at the expense of larger entropy. We now state these results more formally.
 
 \begin{THM}\label{nmext_better_entropy}  There exists a constant $C>0$ such that for all $n, k \in \mathbb{N}$ and any  $\epsilon>0$, with  $k \ge \log(n/\epsilon)2^{2^{C\sqrt{\log \log \log(n/\epsilon)}}}$, there exists an explicit $(k,\epsilon)$-non-malleable extractor $\nmExt:\zo^n \times \zo^d \rightarrow \zo^{m}$, where $d= (\log(n/\epsilon))^32^{(\log \log \log(n/\epsilon))^{O(1)}}, m= \Omega(k)$.
\end{THM}

\begin{THM}\label{nmext_better_seed} There exists a constant $C>0$ such that for all $n, k \in \mathbb{N}$ and any $\epsilon>0$, with  $k \ge (\log(n/\epsilon))^32^{(\log \log \log(n/\epsilon))^C}$, there exists an explicit $(k,\epsilon)$-non-malleable extractor $\nmExt:\zo^n \times \zo^d \rightarrow \zo^{m}$, where $d= \log(n/\epsilon)2^{2^{O(\sqrt{\log \log \log(n/\epsilon)})}}, m= \frac{k}{\log(n/\epsilon)2^{(\log \log \log(n/\epsilon))^{O(1)}}} - O((\log(n/\epsilon))^2)$.
\end{THM}

Table~\ref{nmext_table} summarizes our new non-malleable extractors compared to previous results.

\noindent \textbf{Privacy Amplification}\hspace{0.2cm}Using Theorem $\ref{thm:nmext1}$ and the protocol in \cite{DW09}, we immediately obtain a two-round privacy amplification protocol with optimal entropy loss, for almost all possible security parameters.

\begin{THM}
There exists a constant $C>0$ such that for any security parameter $s$ with $k \geq (s +\log n)2^{C\sqrt{\log (s+ \log n)}}$, there exists an explicit 2-round privacy amplification protocol  for $(n, k)$-sources with  entropy loss $O(\log n+s)$ and communication complexity $(s +\log n)2^{O(\sqrt{\log (s+ \log n)})}$, in the presence of an active adversary.
\end{THM}
 In particular, this gives us two-round privacy amplification protocols with optimal entropy loss for security parameter $s \leq k^{1-\alpha}$ for any constant $\alpha>0$.

Instead if we use the non-malleable extractor from Theorem $\ref{nmext_better_entropy}$, we obtain a  two-round privacy amplification protocol with optimal entropy loss, for even smaller min-entropy (at the expense of larger communication complexity). More formally, we have the following theorem.
\begin{THM}
There exists a constant $C>0$ such that for any security parameter $s$ with $k \geq (s +\log n)2^{2^{C\sqrt{\log \log (s+ \log n)}}}$, there exists an explicit 2-round privacy amplification protocol  for $(n, k)$-sources with  entropy loss $O(\log n+s)$ and communication complexity $(s +\log n)^32^{(\log \log(s+\log n))^{O(1)}}$, in the presence of an active adversary.
\end{THM}
\noindent \textbf{$t$-Non-Malleable Extractors and $2$-Source Extractors}\hspace{0.2cm}Our techniques for constructing non-malleable extractors can be generalized directly to construct $t$-non-malleable extractors (non-malleable extractors with $t$ tampering functions, see Definition $\ref{def_tnmext}$ and Theorem $\ref{nmext_multi}$). Such $t$-non-malleable extractors were used in \cite{CZ15} to construct two-source extractors. With subsequent improvements \cite{Li:2source,Mek:resil}, the best known $2$-source extractor for constant error requires min-entropy $C(\log n)^{10}$, and for polynomially small error requires min-entropy $C (\log n)^{18}$. By plugging in our improved $t$-non-malleable extractor from Theorem $\ref{nmext_multi}$, we obtain two-source extractors that require min-entropy $(\log n)^{8}$ for constant error, and $(\log n)^{14}$ for polynomially small error (see Theorem $\ref{2ext1}$ and Theorem $\ref{2ext2}$). By a well-known connection to Ramsey graphs (see \cite{BRSW12}), the constant error $2$-source extractor implies an explicit $2^{(\log \log n)^8}$-Ramsey graph on $n$ vertices.
 
\noindent \textbf{Multi-Source Extractors}\hspace{0.2cm}Next, we improve the entropy requirement in extractors for a constant number of independent sources. In particular, we give explicit extractors for $O(1)$ sources, each having min-entropy $\log^{1+o(1)}(n)$.  More formally, we have the following theorem.
\begin{THM}\label{multi_thm}There exist constants $C>0, C'>0$ s.t for all $n, k \in \mathbb{N}$ with $k \geq  \log n 2^{C'\sqrt{\log \log(n)}}$ and any constant $\epsilon>0$,\footnote{As in \cite{Coh16b}, the error can actually be slightly sub-constant.} there exists an explicit function $\Ext:(\zo^n)^{C} \rightarrow \zo$, such that if  $\X_1,\ldots,\X_{C}$ are independent $(n, k)$ sources, then
$$|\Ext(\X_1,\ldots,\X_{C}) -\U_1 | \le \epsilon.$$
\end{THM}
\begin{table}[t]\label{nmext_table}
\centering
\begin{tabular}{|m{5 cm}| m{5 cm} | m{5 cm}| @{}m{0pt}@{}} 
\hline
Reference & Min-Entropy & Seed Length &\\ [10pt] \hline
\cite{DW09} (non-constructive) & $> 2m + 2\log(1/\epsilon)+ \log d + 6$ & $>\log(n-k+1)+2\log(1/\epsilon)+5$ &\\ [10pt] \hline \hline
\cite{DLWZ11} & $>n/2$ & $n$ & \\  [10pt]  \hline
\cite{CRS12,Li12a,DY12} & $>n/2$ & $O(\log(n/\epsilon))$& \\ [10pt] \hline
\cite{Li12b} & $0.49n$ & $n$& \\  [10pt] \hline
\cite{CGL15} &$ \Omega((\log(n/\epsilon))^2)$ & $O((\log(n/\epsilon))^2)$& \\  [10pt] \hline
\cite{Coh15nm} &   $\Omega(\log(n/\epsilon)\log((\log n)/\epsilon))$ &      $O(\log(n/\epsilon)\log((\log n)/\epsilon))$& \\ [10pt] \hline
\cite{Coh16a}   &   $\Omega(\log n + (\log(1/\epsilon))^3)$  & $O(\log n + (\log(1/\epsilon))^3)$   &   \\   [10pt] \hline \hline
Theorem $\ref{thm:nmext1}$ & $ \log(n/\epsilon)2^{\Omega(\sqrt{\log \log (n/\epsilon)})}$ &  $\log(n/\epsilon)2^{O(\sqrt{\log \log (n/\epsilon)})}$ & \\  [10pt] \hline
Theorem $\ref{nmext_better_seed}$ &   $\log(n/\epsilon)2^{2^{\Omega(\sqrt{\log \log \log(n/\epsilon)})}} $      & $(\log(n/\epsilon))^{3+o(1)} $   &    \\ [10pt] \hline
Theorem $\ref{nmext_better_entropy}$ &      $(\log(n/\epsilon))^{3+o(1)}$        &       $ \log(n/\epsilon)2^{2^{O(\sqrt{\log \log \log(n/\epsilon)})}}    $  & \\  [10pt]  
\hline
\end{tabular}
\caption{A summary of results on non-malleable extractors}
\end{table}
\subsection{Non-malleable independence preserving merger}
The barrier of $\log^2(1/\eps)$ in seed length and min-entropy requirement of non-malleable extractors, as well as the barrier of $\log^2 n$ in min-entropy requirement of multi-source extractors mainly come from the fact that the previous constructions rely heavily on the ``alternating extraction" based techniques. In \cite{Coh16b}, Cohen and Schulman introduced a new object called \emph{independence preserving merger} ($\ipm$ for short). This is the key component in their construction, which helps them to obtain the $O(1/\delta)+O(1)$ source extractor for min-entropy $k \geq \log^{1+\delta} n$. The construction of the independence preserving merger in \cite{Coh16b} is fairly complicated and takes up a bulk of work. 

A key component in all of our constructions is a generalized, and much more efficient version of the independence preserving merger in \cite{Coh16b}, which we call non-malleable independence preserving merger ($\nipm$ for short). In addition, we believe that our construction of $\nipm$ is simpler than the construction of $\ipm$ in \cite{Coh16b}. We now define this object below.

\begin{define}A $(L,t,d',\epsilon,\epsilon')$-$\nipm: \zo^{Lm} \times \zo^d \rightarrow \zo^{m_1}$ satisfies the following property.  Suppose
\begin{itemize}
\item $\X,\X^1,\ldots,\X^{t}$ are r.v's, each supported on boolean $L\times m$ matrices s.t for any $i \in [L]$, $|\X_i - \U_m| \le \epsilon$,
\item $\{\Y,\Y^1,\ldots,\Y^t\}$ is independent of $\{ \X,\X^1,\ldots,\X^{t}\}$, s.t $\Y,\Y^1,\ldots,\Y^t$ are each supported on $\zo^{d}$ and $H_{\infty}(\Y) \ge d-d'$,
\item there exists an $h \in [L]$ such that $|(\X_h,\X_h^1,\ldots,\X_h^t)-(\U_m,\X_h^1,\ldots,\X_h^t)|\le \epsilon$,
\end{itemize}
then 
\begin{align*}
|(L,t,d',\epsilon,\epsilon')\text{-}\nipm((\X,\Y), (L,t,d',\epsilon,\epsilon')\text{-}\nipm(\X^1,\Y^1),\ldots, (L,t,d',\epsilon,\epsilon')\text{-}\nipm(\X^t,\Y^t)\\ -\U_{m_1},  (L,t,d',\epsilon,\epsilon')\text{-}\nipm(\X^1,\Y^1),\ldots, (L,t,d',\epsilon,\epsilon')\text{-}\nipm(\X^t,\Y^t)| \le \epsilon'.
\end{align*}
\end{define}

We present an explicit construction of an $\nipm$ which requires seed length $d= \log(m/\epsilon)L^{o(1)}$ for the case $t=1$. More formally, we have the following theorem.

\begin{THM}\label{basic_rec_nipm}For all integers $m,L>0$, any $\epsilon>0$,   there exists an explicit $(L,1,0,\epsilon,\epsilon')$-$\nipm:\zo^{mL}\times \zo^{d} \rightarrow \zo^{m'}$, where $d=2^{O(\sqrt{\log L})}\log(m/\epsilon), m'= \frac{m}{2^{\sqrt{\log L}}}-2^{O(\sqrt{\log L})}\log(m/\epsilon)$ and $\epsilon'= O(\epsilon L)$.
\end{THM}

 We have a more general version of the above theorem presented in Section $\ref{sec_nm_rec}$ which works for general $t$. This is crucial for us to obtain our results on $t$-non malleable extractors and extractors for independent sources with near logarithmic min-entropy.
 
Using our $\nipm$, we construct a standard $\ipm$ introduced in the work of Cohen and Schulman \cite{Coh16b}. We first define an $\ipm$.
 
 \begin{define}A $(L,k,t,\epsilon,\epsilon')$-$\ipm: \zo^{Lm} \times \zo^n \rightarrow \zo^{m_1}$ satisfies the following property.  Suppose
\begin{itemize}
\item $\X,\X^1,\ldots,\X^{t}$ are r.v's, each supported on boolean $L\times m$ matrices s.t for any $i \in [L]$, $|\X_i - \U_m| \le \epsilon$,
\item $\Y$ is an $(n,k)$-source, independent of $\{ \X,\X^1,\ldots,\X^{t}\}$.
\item there exists an $h \in [L]$ such that $|(\X_h,\X_h^1,\ldots,\X_h^t)-(\U_m,\X_h^1,\ldots,\X_h^t)|\le \epsilon$,
\end{itemize}
then 
\begin{align*}
|(L,k,t,\epsilon,\epsilon')\text{-}\ipm(\X,\Y), (L,k,t,\epsilon,\epsilon')\text{-}\ipm(\X^1,\Y),\ldots, (L,k,t,\epsilon,\epsilon')\text{-}\nipm(\X^t,\Y)\\ -\U_{m_1},  (L,k,t,\epsilon,\epsilon')\text{-}\ipm(\X^1,\Y),\ldots, (L,k,t,\epsilon,\epsilon')\text{-}\ipm(\X^t,\Y)| \le \epsilon'
\end{align*}
\end{define}

\begin{THM}There exists a constant $C>0$ such that for all integers $m,L>0$, any $\epsilon>0$, and any $k\ge 2^{C\sqrt{\log L}}\log(m/\epsilon)$,   there exists an explicit $(L,k,1,\epsilon,\epsilon')$-$\ipm:\zo^{mL}\times \zo^{n} \rightarrow \zo^{m'}$, with $m'= \frac{1}{2^{\sqrt{\log L}}}(m-O(\log(n/\epsilon))) - 2^{O(\sqrt{\log L})}\log(m/\eps)$ and $\epsilon'= O(\epsilon L)$.
\end{THM}
As in the case of $\nipm$, we in fact construct an $\ipm$ for general $t$. The construction of $\ipm$ from $\nipm$ is relatively straightforward, and using this explicit $\ipm$ we derive our improved results on extractors for independent sources.

We note that there are several important differences between our $\ipm$ and the construction in \cite{Coh16b}. First, we only require that there exists at least \emph{one} ``good" row $\X_h$ in the matrix (i.e., $\X_h$ is uniform even given $\X_h^1,\ldots,\X_h^t$). In contrast, the $\ipm$ in \cite{Coh16b} requires that $0.99$ fraction of the rows are good. Second, the construction  of $\ipm$ in \cite{Coh16b} offers a trade-off between the number of additional sources required and the min-entropy requirement of each source. In particular, they construct an $\ipm$ using $b$ additional sources, each having min-entropy $k = \Omega(L^{1/b}\log(n/\epsilon))$. In contrast, we use just \emph{one} additional source with min-entropy $k=\Omega(L^{o(1)}\log (m/\eps))$, and works as long as $m \ge O(\log(n/\epsilon))+L^{o(1)}\log (m/\eps)$. In typical applications, we will have $m \approx k < n$, so it suffices to set $k = L^{o(1)} \log (n/\eps)$. For all applications in this paper, we will choose $L=O(\log(n/\eps))$ and thus we get $k = \log^{1+o(1)} (n/\eps)$. The fact that our $\ipm$ uses only one additional source improves significantly upon the $\ipm$ in \cite{Coh16b} and is crucial for us to obtain an $O(1)$ source extractor for min-entropy $k=\log^{1+o(1)} n$.

We also present a more involved construction of a $\nipm$ that uses a shorter seed in comparison to the $\nipm$ in Theorem $\ref{basic_rec_nipm}$, but requires matrices with larger rows (i.e., $m$ is required to be larger). More formally, we have the following result.  
\begin{THM}\label{advanced_nipm}For all integers $m,L>0$, any $\epsilon>0$,   there exists an explicit $(L,1,0,\epsilon,\epsilon')$-$\nipm:\zo^{mL}\times \zo^{d} \rightarrow \zo^{m'}$, where $d=2^{O(\sqrt{\log \log L})}\log (m/\eps), m'= \frac{m}{L2^{(\log \log L)^{O(1)}}}-O(L\log(m/\eps))$ and $\epsilon'= 2^{O(\sqrt{\log \log L})} L \eps$.
\end{THM}
We use the $\nipm$ from the above theorem in  obtaining  the non-malleable extractors in Theorem $\ref{nmext_better_entropy}$ and Theorem $\ref{nmext_better_seed}$.
\section{Outline of Constructions}
Here we give an informal and high level description of our constructions. We start with our non-malleable independence preserving merger (NIPM) with a uniform (or high entropy rate) seed. 

\subsection{Non-malleable independence preserving merger with uniform seed}
For simplicity we start by describing the case of only one tampering adversary. Here, we have two correlated random variables $\X=(\X_1,\ldots,\X_{L})$ and $\X'=(\X_1',\ldots,\X_{L}')$, each of them is an $L \times m$ matrix. We have another two correlated random variables $\Y, \Y'$. We assume the following conditions: $(\X, \X')$ is independent of $(\Y, \Y')$, each $\X_i$ is uniform and there exists a $j \in [L]$ such that $\X_j$ is uniform even conditioned on $\X_j'$, and $\Y$ is uniform. Our goal is to construct a function $\nipm$ such that $\nipm(\X, \Y)$ is uniform conditioned on $\nipm(\X', \Y')$, i.e., using $\Y$ we can merge $\X$ into a uniform random string which keeps the independence property over $\X'$ even with a tampered seed $\Y'$.

Our starting point is the following simple observation. Let $(\X, \X')$ be two correlated weak sources and $(\rr, \rr')$ be two correlated random variables such that $(\X, \X')$ is independent of $(\rr, \rr')$. Let $\rr$ be uniform and take any strong seeded extractor $\Ext$, consider $\Z=\Ext(\X, \rr)$ and $\Z'=\Ext(\X', \rr')$. Assume the length of the output of $\Ext$ is small enough. Then $\Z$ is close to uniform given $\Z'$ if either of the following two conditions holds: $\rr$ is uniform given $\rr'$ or $\X$ has sufficient min-entropy conditioned on $\X'$. Indeed, in the first case, we can first fix $\rr'$, and argue that conditioned on this fixing, $\Z'$ is a deterministic function of $X'$. We can now further fix $\Z'$, and conditioned on this fixing, $\X$ still has enough entropy left (since the length of $\Z'$ is small). Note that at this point $\rr$ is still uniform and independent of $X$, thus $\Z=\Ext(\X, \rr)$ is close to uniform given $\Z'$.  In the second case, we can first fix $\X'$, and conditioned on this fixing $\X$ still has enough entropy left. Now since $\Ext$ is a strong extractor, we know that $\Z=\Ext(\X, \rr)$ is close to uniform even given $\rr$. Since we have already fixed $\X'$ and $(\rr, \rr')$ is independent of $(\X, \X')$, this means that $\Z$ is also close to uniform even given $\rr'$ and $\X'$, which gives us $\Z'=\Ext(\X', \rr')$. 

Now we can describe our basic NIPM. The construction is actually simple in the sense that it is essentially an alternating extraction process between $\X$ and $\Y$, except that in each alternation we use a \emph{new} row from $\X$. Specifically, we first take a small slice $\s_1$ from $\X_1$, and apply a strong seeded extractor to obtain $\rr_1=\Ext(\Y, \s_1)$; we then use $\rr_1$ to extract from $\X_2$ and obtain $\s_2=\Ext(\X_2, \rr_1)$. Now we continue and obtain $\rr_2=\Ext(\Y, \s_2)$ and $\s_3=\Ext(\X_3, \rr_2)$... .The final output of our merger will be $\s_{L}=\Ext(\X_{L}, \rr_{L-1})$.

To see why this construction works, first assume that the length of each $\s_i, \rr_i$ is small enough. Let $j$ be the first index in $[L]$ such that $\X_j$ is uniform even conditioned on $\X_j'$. Then, we can fix all the intermediate random variables $\s_1, \s_1', \rr_1, \rr_1', \s_2, \s_2', \rr_2, \rr_2' \ldots \s_{j-1}, \s_{j-1}'$, and conditioned on these fixings we know that: 1. $(\rr_{j-1}, \rr_{j-1}')$ are deterministic functions of $(\Y, \Y')$, and thus independent of $(\X, \X')$; 2. $\rr_{j-1}$ is close to uniform; 3. $\X_j$ still has enough entropy conditioned on $\X_j'$. Now, by the first case we discussed above, this implies that $\s_j$ is close to uniform given $\s_j'$. From this point on, by using the second case we discussed above and an inductive approach, we can argue that for all subsequent $t \geq j$, we have that $\rr_t$ is close to uniform given $\rr_t'$ and $\s_t$ is close to uniform given $\s_t'$. Thus the final output $\s_{L}$ is close to uniform given $\s_{L}'$. 

Note that this construction can work even if $\Y$ is a very weak random source instead of being uniform or having high min-entropy rate. However this basic approach will require the min-entropy of $Y$ to be at least $O(L \log(m/\ep))$, which is pretty large if $L$ is large. We next describe a way to reduce this entropy requirement, in the case where $\Y$ is uniform or has high min-entropy rate.

The idea is that, rather than merging the $L$ rows in one step, we merge them in a sequence of steps, with each step merging all the blocks of some $\ell$ rows. Thus, it will take us roughly $\frac{\log L}{\log \ell}$ steps to merge the entire matrix. Now first assume that $\Y$ is uniform, then in each step we will not use the entire $\Y$ to do the alternating extraction and merging, but just use a \emph{small slice} of $\Y$ for this purpose. That is, we will first take a small slice $\Y_1$ and use this slice to merge $L/\ell$ blocks of $\X$, where each block has $\ell$ rows; we then take another slice $\Y_2$ of $\Y$ and use this slice to merge $L/\ell^2$ new blocks, where each block has $\ell$ rows, and so on. The advantage of this approach is that now the entropy consumed in each merging step is contained in the slice $\Y_i$ (and $\Y_i'$), and won't affect the rest of $\Y$ much.

As we discussed before, we need to make sure that each slice $\Y_i$ has min-entropy $O(\ell \log(m/\ep))$ conditioned on the fixing of all previous $(\Y_j, \Y_j')$. As a result, we need to set $|\Y_{i+1}| \geq 2|\Y_i|+O(\ell \log(m/\ep))$. It suffices to take $|\Y_i|=c^i \ell \log(m/\ep)$ for some constant $c>2$. We know that the whole merging process is going to take roughly $\frac{\log L}{\log \ell}$ steps, so the total length (or min-entropy) of $\Y$ is something like $c^{\frac{\log L}{\log \ell}} \ell \log(m/\ep)$. We just need to choose a proper $\ell$ to minimize this quantity. A simple calculation shows that the best $\ell$ is roughly such that $\log \ell =\sqrt{\log L}$, which gives us a seed length of $2^{O(\sqrt{\log L})} \log(m/\ep)$. This gives us the $\nipm$ in Theorem~\ref{basic_rec_nipm}.

It is not difficult to see that this argument also extends to the case where $\Y$ is not perfectly uniform but has high min-entropy rate (e.g., $1-o(1)$) where we can still start with a small slice of $\Y$, and the case where we have $t+1$ correlated matrices $\X, \X^1, \ldots, \X^t$ and $t$ tampered seeds $\Y^1, \ldots, \Y^t$ of $\Y$.

\subsection{Non-malleable extractor with almost optimal seed}
The  NIPM in Theorem $\ref{basic_rec_nipm}$ is already enough to yield our construction of a non-malleable extractor with almost optimal seed length. Specifically, given an $(n, k)$ source $\X$, an independent seed $\Y$, and a tampered seed $\Y'$, we follow the approach of one of the authors' previous work \cite{CGL15} by first obtaining an advice of length $L=O(\log (n/\ep))$. Let the advice generated by $(\X, \Y)$ be $S$ and the advice generated by $\X, \Y'$ be $\s'$. We have that with probability $1-\ep$, $\s \neq \s'$. Further, conditioned on $(\s, \s')$ and some other random variables, we have that $\X$ is still independent of $(\Y, \Y')$ and $\Y$ has high min-entropy rate. 

Now we take a small slice $\Y_1$ of $\Y$, and use $\X$ and $\Y_1$ to generate a random matrix $\V$ with $L$ rows, where the $i$'th row is obtained by doing a flip-flop alternating extraction (introduced in \cite{Coh15}) using the $i$'th bit of $S$. Similarly a matrix $\V'$ is generated using $\X'$ and $\Y_1$. The flip-flop alternating extraction guarantees that each row in $\V$ is close to uniform, and moreover if the $i$'th bit of $S$ and $S'$ are different, then $\V_i$ is close to uniform even given $\V_i'$. Note that conditioned on the fixing of $(\Y_1, \Y_1')$, we have that $(\V, \V')$ are deterministic functions of $\X$, and are thus independent of $(\Y, \Y')$. Furthermore $\Y$ still has high min-entropy rate.

At this point we can just use our NIPM and $\Y$ to merge $\V$ into a uniform string $\Z$, which is guaranteed to be close to uniform given $\Z'$ (obtained from $(\V', \Y')$). The seed length of $\Y$ will be $O(\log(n/\ep))+2^{O(\sqrt{\log \log(n/\ep)})} \log(k/\ep)$. A careful analysis shows that the final error will be $O(\ep \log(n/\ep))$. Thus we need to set the error parameter $\ep$ slightly smaller in order to achieve a desired final error $\ep'$, but that does not affect the seed length much. Altogether this gives us a seed length and entropy requirement of $2^{O(\sqrt{\log \log(n/\ep)})} \log(n/\ep)$, as in Theorem $\ref{thm:nmext1}$. 


\subsection{Further improvements in various aspects}
We can further improve the non-malleable independence preserving merger and non-malleable extractor in various aspects. For this purpose, we observe that our NIPM starts with a basic merger for $\ell$ rows and then use roughly $\frac{\log L}{\log \ell}$ steps to merge the entire $L$ rows. If the basic merger uses a seed length of $d$, then the whole merger roughly uses seed length $c^{\frac{\log L}{\log \ell}} d$. In our basic and simple merger, we have $d=O(\ell \log(m/\ep))$. However, now that we have our improved NIPM, we can certainly use the more involved construction to replace the basic merger, where we only need seed length $d=2^{O(\sqrt{\log \ell})} \log(m/\ep)$. Now we can choose another $\ell$ to optimize $c^{\frac{\log L}{\log \ell}} d$, which roughly gives $\log \ell =\log^{2/3} L$ and the new seed length is $2^{O(\log^{1/3} L)} \log(m/\ep)$. We can now again use this merger to replace the basic merger. By doing this recursively, we can get smaller and smaller seed length. On the other hand, the entropy requirement becomes larger. We can also switch the roles of the seed and source, and achieve smaller entropy requirement at the price of a larger seed. Eventually, we can get $d=2^{O(\sqrt{\log \log L})}\log (m/\eps)$ and $m=O(L^22^{(\log \log L)^{O(1)}})\log (m/\eps)$, or vice versa. This gives us Theorem $\ref{advanced_nipm}$. Applying these $\nipm$s to non-malleable extractors as outlined above, we get Theorem~\ref{nmext_better_entropy} and Theorem~\ref{nmext_better_seed}.

\subsection{Independence preserving merger with weak random seed}
We now use our NIPM from Theorem $\ref{basic_rec_nipm}$ to construct a standard independence preserving merger with weak random seed, an object introduced in \cite{Coh16b}. Suppose we are given $(\X, \X')$ as described above and an independent random variable $\Y$. Here $\Y$ can be a very weak source, so our first step is to convert it to a uniform (or high min-entropy rate) seed. 

To do this, our observation is that since we know that each row in $\X$ is uniform, we can just take a small slice $\W$ of the first row $\X_1$, and apply a strong seeded extractor to $\Y$ to obtain $\Z=\Ext(\Y, \W)$, which is guaranteed to be close to uniform. However by doing this we also created a correlated $\Z'=\Ext(\Y, \W')$ where $\W'$ is a slice of $\X_1'$. Note that conditioned on the fixing of $(\W, \W')$ we have $(\X, \X')$ is independent of $(\Z, \Z')$, each row of $\X$ still has high min-entropy, and the ``good" row $\X_j$ still has high min-entropy even given $\X_j'$. We now take a small slice $\V$ of $\Z$, and use it to extract from each row of $\X$ to obtain another matrix $\overline{\X}$. Similarly we also have a slice $\V'$ from $\Z'$ and obtain $\overline{\X'}$. We can now argue that conditioned on the fixing of $(\V, \V')$, $(\overline{\X}, \overline{\X'})$ is independent of $(\Z, \Z')$, each row of $\overline{\X}$ is close to uniform, and the ``good" row $\overline{\X}_j$ is close to uniform even given $\overline{\X'}_j$. Moreover $\Z$ still has high min-entropy rate. 

Thus, we have reduced this case to the case of an independence preserving merger with a tampered high min-entropy rate seed. We can therefore apply our NIPM to finish the construction. It is also not difficult to see that our construction can be extended to the case where we have $\X, \X^1, \ldots, \X^t$ instead of having just $\X$ and $\X'$.

\subsection{Improved multi-source extractor}
We can now apply our independence preserving merger with weak random seed to improve the multi-source extractor construction in \cite{Coh16b}. Our construction follows the framework of that in \cite{Coh16b}, except that we replace their independence preserving merger with ours. Essentially, the key step in the construction of \cite{Coh16b}, and the only step which takes $O(1/\delta)$ independent $(n, \log^{1+\delta} n)$ sources (if we only aim at achieving constant or slightly sub-constant error) is to merge a matrix with $O(\log n)$ rows using $(n, \log^{1+\delta} n)$ sources. For this purpose and since the error of the merger needs to be $1/\poly(n)$, the independence preserving merger in \cite{Coh16b} uses two additional sources in each step to reduce the number of rows by a factor of $\log^{\delta} n$. Thus altogether it takes $2/\delta$ sources. Our merger as described above, in contrast, only requires \emph{one} extra independent source with min-entropy at least $O(\log n)+2^{O(\sqrt{\log \log n})} \log n=\log^{1+o(1)} n$. Therefore, we obtain a multi-source extractor for an absolute constant number of $(n, \log^{1+o(1)} n)$ sources, which outputs one bit with constant (or slightly sub-constant) error. 

The improved two-source extractors are obtained directly by plugging in our improved $t$-non-malleable extractors to the constructions in \cite{CZ15,Mek:resil}.

\paragraph{Organization}
We introduce some preliminaries in Section $\ref{section:prelims}$. We present our constructions of non-malleable independence preserving mergers and independence preserving mergers in Section $\ref{sec:mergers}$. We use Section $\ref{sec:opt_nm}$ to present  the construction of our almost-optimal non-malleable extractor. We present improved constructions of $t$-non-malleable extractors and applications to $2$-source extractors in Section $\ref{sec_t}$. We use Section $\ref{sec:involved}$ to present the $\nipm$ construction of Theorem $\ref{advanced_nipm}$, and the non-malleable extractor constructions of Theorem $\ref{nmext_better_entropy}$ and Theorem $\ref{nmext_better_seed}$. We present our results on multi-source extractors in Section $\ref{multi_source}$.

%% file: preliminaries.tex
\section{Preliminaries} \label{section:prelims}
 We use $\U_m$ to denote the uniform distribution on $\{0,1 \}^m$.  \newline For any integer $t>0$, $[t]$ denotes the set $\{1,\ldots,t \}$.\newline For a string $y$ of length $n$, and any subset $S \subseteq [n]$, we use $y_S$ to denote the projection of $y$ to the coordinates indexed by $S$. \newline For a string $y$ of length $m$, define the string $\slice(y,w)$ to be the prefix of length $w$ of $y$. \newline We use bold capital letters for random variables and  samples as the corresponding small letter, e.g., $\X$ is a random variable, with $x$ being a sample of $\X$.

\subsection{Conditional Min-Entropy}
\begin{define} The average conditional min-entropy of a source $\X$ given a random variable $\W$ is defined as $$ \widetilde{H}_{\infty}(\X|\W) = -\log \l( \E_{w \sim W}\l[\max_{x} \Pr[\X=x | \W=w] \r] \r) = - \log \l(\E\l[ 2^{-H_{\infty}(\X|\W=w)} \r]\r).$$
\end{define}
We recall some results on conditional min-entropy from the work of Dodis et al.\ \cite{DORS08}.
\begin{lemma}[\cite{DORS08}] For any $\epsilon>0$, $\pr_{w \sim \W}\l[H_{\infty}(\X|\W=w) \ge \widetilde{H}_{\infty}(\X|\W)-\log(1/\epsilon)\r] \ge 1- \epsilon$.
\end{lemma}
\begin{lemma}[\cite{DORS08}]\label{lem:entropy_loss} If a random variable $\Y$ has support of size $2^\ell$, then $\widetilde{H}_{\infty}(\X|\Y) \ge H_{\infty}(\X) - \ell$.
\end{lemma}
We require extractors that can extract  uniform bits when the source only has sufficient conditional min-entropy. 
\begin{define} A $(k,\epsilon)$-seeded average case seeded extractor $\Ext:\{ 0,1\}^n \times \{ 0,1\}^d \rightarrow \{ 0,1\}^m$ for min-entropy $k$ and error $\epsilon$ satisfies the following property:  For any source $\X$ and any arbitrary random variable $\Z$ with $\tilde{H}_{\infty}(\X|\Z)\ge k$, $$\Ext(\X,\U_d),\Z \approx_{\epsilon} \U_m, \Z.$$ 
\end{define}
It was shown in \cite{DORS08} that any seeded extractor is also an average case extractor.
\begin{lemma}[\cite{DORS08}]\label{lem:cond_ext} For any $\delta>0$, if $\Ext$ is a  $(k,\epsilon)$-seeded extractor, then it is also a  $(k+\log(1/\delta),\epsilon+\delta)$-seeded average case extractor.
\end{lemma}

\subsection{Some Probability Lemmas}
The following result on min-entropy was proved by Maurer and  Wolf \cite{MW07}.
\begin{lemma}\label{lemma:entropy_loss_1} Let $\X,\Y$ be random variables such that the random variable $\Y$ takes at $\ell$ values. Then 
\begin{align*}
 \pr_{y \sim \Y}\l[ H_{\infty}(\X| \Y = y) \ge H_{\infty}(\X) - \log \ell -\log\l(\frac{1}{\epsilon}\r)\r] > 1-\epsilon.
 \end{align*}
\end{lemma} 
\begin{lemma}[\cite{BIW}]\label{sum_rv} Let $\X_1,\ldots,\X_{\ell}$ be independent random variables on $\zo^m$ such that $|\X_i-\U_m| \le \epsilon$. Then, $|\sum_{i=1}^{\ell} \X_i - \U_m| \le \epsilon^{\ell}$.
\end{lemma}
\subsection{Seeded Extractors}
We  use optimal constructions of strong-seeded extractors.
\begin{thm}[\cite{GUV09}]\label{guv} For any constant $\alpha>0$, and all integers $n,k>0$ there exists a polynomial time computable  strong-seeded extractor $\Ext: \{ 0,1\}^n \times \{ 0,1\}^d \rightarrow \{ 0,1\}^m$   with $d = c_{\ref{guv}}(\log n + \log (1/\epsilon))$ and $m = (1-\alpha)k$.
\end{thm}

%% file: nm_extractor.tex
\section{Non-Malleable Independence Preserving Mergers}\label{sec:mergers}
In this section we present our constructions of non-malleable independence preserving mergers. We gradually develop the ideas to build our final NIPM, starting with a simpler construction which is then used as a building block in the more involved construction.  In our proofs, we repeatedly condition on random variables with small support and account for entropy loss by implicitly using Lemma $\ref{lem:entropy_loss}$.
\subsection{$\ell$-Non-Malleable Independence Preserving Merger}\label{basic_merger}
In this section, we construct an explicit function $\nipm$ that uses a (weak) seed $\Y$ to merge $\ell$ correlated r.v's $\X_1,\ldots,\X_{\ell},\X_1',\ldots,\X_{\ell}'$ in a way such that  such that if for some $i$, $\X_i|\X_i'$'s is close to uniform on average, then this property is transferred to the output of $\nipm$. As discussed in the introduction, a recent work by Cohen and Schulman \cite{Coh16b} introduces a similar object in the context of constructing multi-source extractors with nearly logarithmic min-entropy. However there are some important differences. To carry out the independence preserving merging, \cite{Coh16b} uses access to multiple independent sources which are themselves not tampered. Here we allow access to an  independent  weak seed $\Y$ which is further subject to being tampered ($\Y'$ being the tampered seed). 

The following is the main result of this section.
\begin{thm}\label{thm:nipm}
There exist constants $c_{\ref{thm:nipm}},c'_{\ref{thm:nipm}}>0$ such that for all integers $m,d,k_1,\ell>0$  and any $\epsilon>0$,  with $m\ge d  \ge k_1 > c_{\ref{thm:nipm}} \ell \log(n/\epsilon)$,  there exists an explicit function $\ell$-$\nipm:(\zo^{m})^{\ell} \times \zo^d \rightarrow \zo^{m_1}$, $m_1=0.9(m-c_{\ref{thm:nipm}}\ell\log(m/\epsilon))$, such that if the following conditions hold:

\begin{itemize}
\item $\X_1,\ldots,\X_{\ell}$ are r.v's s.t for all $i \in [\ell]$, $|\X_i - \U_m| \le \epsilon_1$, and  $\X_1',\ldots,\X_{\ell}'$ are r.v's with each $\X_i'$  supported on $\zo^m$. 
\item $\{\Y,\Y'\}$ is independent of $\{ \X_1,\ldots,\X_t,\X_1',\ldots,\X_t'\}$, s.t the r.v's $\Y,\Y'$ are both supported on $\zo^{d}$ and $H_{\infty}(\Y) \ge k_1 $.
\item there exists an $h \in [t]$ such that $|(\X_h,\X_h')-(\U_m,\X_h')| \le \epsilon$,
\end{itemize}
then 
\begin{align*}
| \ell\text{-}\nipm((\X_1,\ldots,\X_{\ell}),\Y), \ell\text{-}\nipm((\X_1',\ldots,\X_{\ell}'),\Y'),\Y,\Y' \\ -\U_{m_1}, \ell\text{-}\nipm((\X_1',\ldots,\X_{\ell}'),\Y'),\Y,\Y'| \le c_{\ref{thm:nipm}}'\ell \epsilon
\end{align*}
\end{thm}

Our construction of $\nipm$ uses the method of alternating extraction and extends it in a new way. Briefly we recall the method of alternating extraction which was introduced by Dziembowski and Pietrzak \cite{DP07}, and has been useful in a variety of extractor constructions \cite{DW09,Li13b,Li15c,Coh15,CGL15,Li:affine,CL15,Coh15nm,Coh16a}.

\textbf{Alternating Extraction}  Assume that there are two parties, Quentin with a source $\Q$  and Wendy with a source $\w$. The alternating extraction protocol is an interactive process between Quentin and Wendy, and starts off with Quentin sending the seed $\s_0$ to Wendy.  Wendy uses $\s_0$ and a strong-seeded extractor $\Ext_w$ to extract a seed $\rr_1$ using $\w$, and sends $\rr_1$ back to Quentin. This constitutes a round of the alternating extraction protocol. In the next round, Quentin uses a strong extractor $\Ext_q$ to extract a seed $\s_2$ from $\Q$ using $\s_1$, and sends it to Wendy and so on. The protocol is run for $h$ steps, where $h$ is an input parameter. Thus,  the following sequence of r.v's is generated: $$ \s_1=\slice(\Q,d), \rr_1 = \Ext_{w}(\W,\s_1), \s_1 = \Ext_{q}(\Q,\rr_1),\ldots,\s_{u} = \Ext_{q}(\Q,\rr_{h-1}).$$
 
\textbf{$\ell$-Alternating Extraction} We extend the above technique by letting Quentin have access to $\ell$ sources $\Q_1,\ldots,\Q_{\ell}$ (instead of just $\Q$)  and $\ell$ strong-seeded extractors $\{\Ext_{q,i}: i \in [\ell]\}$ such that in the $i$'th round of the protocol, he uses $\Q_i$ to produce the r.v $\s_i=\Ext_{q,i}(\Q_i,\rr_{i})$. More formally, the following sequence of r.v's is  generated: $ \s_1=\slice(\Q_1,d), \rr_1 = \Ext_{w}(\W,\s_1), \s_2 = \Ext_{q,2}(\Q_2,\rr_1),\ldots,\rr_{\ell-1}=\Ext_{w}(\Q_{\ell-1},\s_{\ell-1}),\s_{\ell} = \Ext_{q,\ell}(\Q_{\ell},\rr_{\ell})$. Define the look-ahead extractor $$\ell\text{-}\laExt((\Q_1,\ldots,\Q_{\ell}),\W)=\s_{\ell}.$$

We are now ready to prove Theorem $\ref{thm:nipm}$.

\begin{proof}[Proof of Theorem $\ref{thm:nipm}$] We instantiate the $\ell$-look-ahead extractor described above with the following strong seeded extractors: Let $\Ext_1: \zo^m \times \zo^{d_1} \rightarrow \zo^{d_1}$, $\Ext_2: \zo^d \times \zo^{d_1} \rightarrow \zo^{d_1}$ and $\Ext_3:\{ 0,1\}^{m} \times \{ 0,1\}^{d_1} \rightarrow \{ 0,1\}^{m_1}$ be explicit strong-seeded from Theorem $\ref{guv}$ designed to extract from min-entropy $m/2,k_1/4, m-c_{\ref{thm:nipm}}\ell \log(m/\epsilon)$ respectively, each  with error $\epsilon$. Thus $d_1=c_{\ref{guv}}\log(m/\epsilon)$. 

We think of each $\X_i$ being uniform, and add back an error $\epsilon_1 \ell$ in the end.

For each $i\in[\ell-1]$, let $\Ext_{q,i}=\Ext_1$, $\Ext_{q,\ell}=\Ext_3$ and $\Ext_{w}=\Ext_2$.

Define $$\nipm((\X_1,\ldots,\X_{\ell}),\Y)=\laExt((\X_1,\ldots,\X_{\ell}),\Y).$$ For any random variable $\V=f((\X_1,\ldots,\X_{\ell}),\Y)$ (where $f$ is an arbitrary deterministic function), let $\V^{\prime}=f((\X_1',\ldots,\X_{\ell}'),\Y')$. 

We first prove the following claim.

\begin{claim}\label{cl:indep1}
For any $j\in [h-1]$, conditioned on the r.v's $\{\s_i: i \in [j-1] \}, \{ \s_i': i \in [j-1]\}, \{ \rr_i : i \in [j-1]\}, \{ \rr_i' : i \in [j-1]\}$ the following hold:
\begin{itemize}
\item $\s_{j}$ is $2(j-1)\epsilon$-close to $\U_{d_1}$,
\item $\s_{j},\s_{j}'$ are deterministic functions of $\{\X_j,\X_j' \}$,
\item for each $i \in [t]$, $\X_i$ has average conditional min-entropy at least $m-2(j-1)d_1-\log(1/\epsilon)$,
\item $\Y$ has average conditional min-entropy at least $k_1-2(j-1)d_1-\log(1/\epsilon)$,
\item $\{ \X_1,\ldots,\X_{\ell},\X_1',\ldots,\X_{\ell}'\}$ is independent of $\{ \Y,\Y^{\prime}\}$.
\end{itemize}
Further, conditioned on the r.v's $\{\s_i: i \in [j] \}, \{ \s_i': i \in [j]\}, \{ \rr_i : i \in [j-1]\}, \{ \rr_i' : i \in [j-1]\}$ the following hold:
\begin{itemize}
\item $\rr_{j}$ is $(2j-1)\epsilon$-close to $\U_d$,
\item $\rr_{j},\rr_{j}'$ are deterministic functions of $\{\Y,\Y' \}$,
\item for any $i \in [\ell]$, $\X_i$ has average conditional min-entropy at least $m-2jd_1-\log(1/\epsilon)$,
\item $\Y$ has average conditional min-entropy at least $k_1-2(j-1)d_1-\log(1/\epsilon)$,
\item $\{ \X_1,\ldots,\X_{\ell},\X_1',\ldots,\X_{\ell}'\}$ is independent of $\{ \Y,\Y^{\prime}\}$.
\end{itemize}
\end{claim}
\begin{proof}
We prove the above  by induction on $j$. The base case when $j=1$ is direct. Thus suppose $j>1$.  Fix the r.v's $\{\s_i: i \in [j-1] \}, \{ \s_i': i \in [j-1]\}, \{ \rr_i : i \in [j-2]\}, \{ \rr_i' : i \in [j-2]\}$. Using inductive hypothesis, it follows that 
\begin{itemize}
\item $\rr_{j-1}$ is $(2j-3)\epsilon$-close to $\U_d$,
\item $\rr_{j-1},\rr_{j-1}'$ are deterministic functions of $\{\Y,\Y' \}$,
\item for any $i \in [t]$, $\X_i$ has average conditional min-entropy at least $m-2(j-1)d_1-\log(1/\epsilon)$,
\item $\Y$ has average conditional min-entropy at least $k_1-2(j-2)d_1-\log(1/\epsilon)$,
\item $\{ \X_1,\ldots,\X_{\ell},\X_1',\ldots,\X_{\ell}'\}$ is independent of $\{ \Y,\Y^{\prime}\}$.
\end{itemize}

Now since $\s_{j}=\Ext_1(\X_j,\rr_{j-1})$, it follows that $\s_{j}$ is $2(j-1)\epsilon$-close to $\U_{d_1}$ on average conditioned on $\rr_{j-1}$. We thus fix $\rr_{j-1}$. Further, we  also fix $\rr_{j-1}'$ without affecting the distribution of $\s_{j}$. Thus $\s_j,\s_j'$ are now a deterministic function of $\X_j,\X_j'$. It follows that after these fixings,  the average conditional  min-entropy of $\Y$ is at least $k_1-2(j-2)d_1-\log(1/\epsilon)-2d_1=k_1-2(j-1)d_1-\log(1/\epsilon)$.

Next, we have $\rr_{j}=\Ext_2(\Y,\s_{j})$, and thus fixing $\s_{j}$, it follows that $\rr_j$ is $(2j-1)\epsilon$-close to uniform on average. Further, since $\rr_j$ is now a deterministic function of $\Y$, we fix $\s_j'$. As a result of these fixings, each $\X_i$ loses conditional min-entropy at most $2d_1$ on average. Since at each point, we either fix a r.v that is a deterministic function of either $\{\X_1,\ldots,\X_{\ell},\X_1',\ldots,\X_{\ell}\}$ or $\{ \Y,\Y'\}$ it follows that  $\{ \X_1,\ldots,\X_{\ell},\X_1',\ldots,\X_{\ell}'\}$ remain  independent of $\{ \Y,\Y^{\prime}\}$. This completes the inductive step, and hence the proof follows.
\end{proof}
We now proceed to prove the following claim.
\begin{claim}\label{cl:indep2}Conditioned on the r.v's $\{\s_i: i \in [h-1] \}, \{ \s_i': i \in [h]\}, \{ \rr_i : i \in [h-1]\}, \{ \rr_i' : i \in [h]\}$ the following hold:
\begin{itemize}
\item $\s_{h}$ is $2(h-1)\epsilon$-close to $\U_d$,
\item $\s_{h}$ is a deterministic function of $\X_h$,
\item for each $i \in [t]$, $\X_i$ has average conditional min-entropy at least $m-2hd_1-\log(1/\epsilon)$,
\item $\Y$ has average conditional min-entropy at least $k_1-2hd_1-\log(1/\epsilon)$,
\item $\{ \X_1,\ldots,\X_{\ell},\X_1',\ldots,\X_{\ell}'\}$ is independent of $\{ \Y,\Y^{\prime}\}$.
\end{itemize}
\end{claim}
\begin{proof}We fix the r.v's $\{\s_i: i \in [h-1] \}, \{ \s_i': i \in [h-1]\}, \{ \rr_i : i \in [h-2]\}, \{ \rr_i' : i \in [h-2]\}$, and using Claim  $\ref{cl:indep1}$ the following hold:
\begin{itemize}
\item $\rr_{h-1}$ is $(2h-3)\epsilon$-close to $\U_d$,
\item $\rr_{h-1},\rr_{h-1}'$ are deterministic functions of $\{\Y,\Y' \}$,
\item for any $i \in [t]$, $\X_i$ has average conditional min-entropy at least $m-2(h-1)d_1-\log(1/\epsilon)$,
\item $\Y$ has average conditional min-entropy at least $k_1-2(h-2)d_1-\log(1/\epsilon)$,
\item $\{ \X_1,\ldots,\X_{\ell},\X_1',\ldots,\X_{\ell}'\}$ is independent of $\{ \Y,\Y^{\prime}\}$.
\end{itemize}
Next we claim that $\X_h$ has average conditional min-entropy at least $m-2(h-1)d_1-\log(1/\epsilon)$ even after fixing $\X_h'$. We know that before  fixings any other r.v, we have $\X_h|\X_h'$ is $\epsilon$-close to uniform on average. Since while computing the average conditional min-entropy, the order of fixing does not matter,  we can as well think of first fixing of $\X_h'$ and then fixing the r.v's $\{\s_i: i \in [h-1] \}, \{ \s_i': i \in [h-1]\}, \{ \rr_i : i \in [h-2]\}, \{ \rr_i' : i \in [h-2]\}$. Thus, it follows  that the average conditional min-entropy of $\X_h$ is at least $m-2(h-1)d_1-\log(1/\epsilon)$. 

We now show that even after fixing the r.v's $\X_h',\rr_{h-1},\rr_{h-1}'$,  the r.v $\s_h$ is $2(h-1)\epsilon$-close to uniform on average. Fix $\X_h'$ and by the above argument  $\X_h$ has average conditional min-entropy at least $m-2(h-1)d_1-\log(1/\epsilon)$. Since $\s_h=\Ext_1(\X_h,\rr_{h-1})$, it follows that $\s_h$ is $2(h-1)\epsilon$-close to uniform on average even conditioned on $\rr_{h-1}$. We fix $\rr_{h-1}$, and thus $\s_h$ is a deterministic function of $\X_h$. Note that $\s_{h}'=\Ext_1(\X_h',\rr_{h-1}')$ is now a deterministic function of $\rr_h'$ (and thus $\Y'$). Thus, we can fix $\rr_h'$ (which also fixes $\s_h')$ without affecting the distribution of $\s_h$. 

Observe that after the r.v's $\rr_{h-1},\rr_{h-1}'$ are fixed, $\s_h'$ is a deterministic function of $\X_h'$. We only fix $\s_h'$ and do not fix $\X_h'$, and note that $\s_h$ is still $2(h-1)\epsilon$-close to uniform. Further after these fixings, each $\X_i$ has average conditional min-entropy at least $m-2hd_1-\log(1/\epsilon)$, and $\Y$ has average conditional min-entropy at least $k_1-2hd_1-\log(1/\epsilon)$.
\end{proof}
By our construction of $\nipm$, Theorem $\ref{thm:nipm}$ is direct from the following claim.
\begin{claim}\label{cl:indep3}  For any $j\in [h,\ell]$, conditioned on the r.v's $\{\s_i: i \in [j-1] \}, \{ \s_i': i \in [j]\}, \{ \rr_i : i \in [j-1]\}, \{ \rr_i' : i \in [j]\}$ the following hold:
\begin{itemize}
\item $\s_{j}$ is $2(j-1)\epsilon$-close to $\U_d$,
\item $\s_{j}$ is a deterministic function of $\X_j$
\item for each $i \in [\ell]$, $\X_i$ has average conditional min-entropy at least $m-2jd_1-\log(1/\epsilon)$,
\item $\Y$ has average conditional min-entropy at least $k_1-2jd_1-\log(1/\epsilon)$,
\item $\{ \X_1,\ldots,\X_{\ell},\X_1',\ldots,\X_{\ell}'\}$ is independent of $\{ \Y,\Y^{\prime}\}$.
\end{itemize}
Further, conditioned on the r.v's $\{\s_i: i \in [j] \}, \{ \s_i': i \in [j+1]\}, \{ \rr_i : i \in [j-1]\}, \{ \rr_i' : i \in [j]\}$ the following hold:
\begin{itemize}
\item $\rr_{j}$ is $(2j-1)\epsilon$-close to $\U_d$,
\item $\rr_{j}$ is a  deterministic function of $\Y$,
\item for any $i \in [\ell]$, $\X_i$ has average conditional min-entropy at least $m-2(j+1)d_1-\log(1/\epsilon)$,
\item $\Y$ has average conditional min-entropy at least $k_1-2(j+1)d_1-\log(1/\epsilon)$,
\item $\{ \X_1,\ldots,\X_{\ell},\X_1',\ldots,\X_{\ell}'\}$ is independent of $\{ \Y,\Y^{\prime}\}$.
\end{itemize}
\end{claim}
\begin{proof}We prove this by induction on $j$. For the base case, when $j=h$, fix the r.v's $\{\s_i: i \in [h-1] \}, \{ \s_i': i \in [h]\}, \{ \rr_i : i \in [h-1]\}, \{ \rr_i' : i \in [h]\}$. Using Claim $\ref{cl:indep2}$, it follows that 
\begin{itemize}
\item $\s_{h}$ is $2(h-1)\epsilon$-close to $\U_d$,
\item $\s_{h}$ is a deterministic function of $\X_h$,
\item for each $i \in [\ell]$, $\X_i$ has average conditional min-entropy at least $m-2hd_1-\log(1/\epsilon)$,
\item $\Y$ has average conditional min-entropy at least $k_1-2hd_1-\log(1/\epsilon)$,
\item $\{ \X_1,\ldots,\X_{\ell},\X_1',\ldots,\X_{\ell}'\}$ is independent of $\{ \Y,\Y^{\prime}\}$.
\end{itemize}
Noting that $\rr_h=Ext_2(\Y,\s_h)$, we fix $\s_h$ and $\rr_h$ is $2h\epsilon$-uniform on average after this fixing. We note that $\rr_h$ is now a deterministic function of $\Y$. Since $\rr_h'$ is fixed, $\s_{h+1}'$ is a deterministic function of $\X_{h+1}'$, and we fix it without affecting the distribution of $\rr_h$. The average conditional min-entropy of each $\X_i$ after these fixings is at least $m-2(h+1)d_1-\log(1/\epsilon)$. Further, we note that our fixings preserve the independence between $\{ \X_1,\ldots,\X_{\ell},\X_1',\ldots,\X_{\ell}'\}$ and $\{ \Y,\Y^{\prime}\}$. This completes the proof of the base case.

Now suppose $j>h$.  Fix the r.v's $\{\s_i: i \in [j-1] \}, \{ \s_i': i \in [j]\}, \{ \rr_i : i \in [j-2]\}, \{ \rr_i' : i \in [j-1]\}$.  Using inductive hypothesis, it follows that 
\begin{itemize}
\item $\rr_{j-1}$ is $(2j-3)\epsilon$-close to $\U_d$,
\item $\rr_{j-1}$ is a  deterministic function of $\Y$,
\item for any $i \in [t]$, $\X_i$ has average conditional min-entropy at least $m-2jd_1-\log(1/\epsilon)$,
\item $\Y$ has average conditional min-entropy at least $k_1-2jd_1-\log(1/\epsilon)$,
\item $\{ \X_1,\ldots,\X_{\ell},\X_1',\ldots,\X_{\ell}'\}$ is independent of $\{ \Y,\Y^{\prime}\}$.
\end{itemize}
Using the fact that $\s_{j}=\Ext_1(\X_j,\rr_{j-1})$, we fix $\rr_{j-1}$ and $\s_j$ is $(2j-2)\epsilon$-close to uniform on average after this fixing. Further, $\s_j$ is a deterministic function of $\X_j$. Since $\s_j'$ is fixed, it follows that $\rr_{j}'$ is a deterministic function of $\Y$ and we fix it without affecting the distribution of $\s_j$. We note that after these fixings, $\Y$ has average conditional min-entropy at least $k_1-2(j+1)d_1-\log(1/\epsilon)$. Further, we note that our fixings preserve the independence between $\{ \X_1,\ldots,\X_{\ell},\X_1',\ldots,\X_{\ell}'\}$ and $\{ \Y,\Y^{\prime}\}$. 

Now, we fix $\s_j$ and it follows that $\rr_j$ is a deterministic function of $\Y$ and is $(2j-1)\epsilon$-close to uniform on average. Further, since $\rr_j'$ is fixed, it follows that $\s_{j+1}'$ is a deterministic function of $\X_{j+1}$ and we fix it without affecting the distribution of $\rr_j$. The average conditional min-entropy of each $\X_i$ after these fixings is at least $m-2(j+1)d_1-\log(1/\epsilon)$. Further, we note that our fixings preserve the independence between $\{ \X_1,\ldots,\X_{\ell},\X_1',\ldots,\X_{\ell}'\}$ and $\{ \Y,\Y^{\prime}\}$. 

This completes the proof of  inductive step, and hence the claim follows.
\end{proof}

\end{proof}

\subsection{$(\ell,t)$-Non-Malleable Independence Preserving Merger}\label{sec:t-nipm}
In this section, we generalize the construction of $\nipm$ from Section $\ref{sec:mergers}$ to handle multiple adversaries.

We first introduce some notation. For a random variable $\V$ supported on $a \times b$ matrices, we use $\V_i$ to denote the random variable corresponding to the $i$'th row of $\V$. Our main result in this section is the following theorem.
\begin{thm}\label{thm:t-nipm}
There exists constant $c_{\ref{thm:t-nipm}},c_{\ref{thm:t-nipm}}'>0$ such that for all integers $m,d,k_1,\ell,t>0$  and any $\epsilon>0$,  with $m\ge d \ge k_1 >c_{\ref{thm:t-nipm}} (t+1) \ell \log(m/\epsilon)$,  there exists an explicit function $t$-$\nipm:\zo^{m\ell}\times \zo^{d} \rightarrow \zo^{m_1}$, $m_1=\frac{0.9}{t}(m-c_{\ref{thm:t-nipm}} (t+1) \ell \log(m/\epsilon))$ such that if the following conditions hold:

\begin{itemize}
\item $\X,\X^1,\ldots,\X^{t}$ are r.v's, each  supported on boolean $\ell \times m$ matrices s.t for any $i \in [\ell]$, $|\X_i - \U_m| \le \epsilon$,
\item $\{\Y,\Y^1,\ldots,\Y^t\}$ is independent of $\{ \X,\X^1,\ldots,\X^{t}\}$, s.t $\Y,\Y^1,\ldots,\Y^t$ are each supported on $\zo^{d}$ and $H_{\infty}(\Y) \ge k_1$.
\item there exists an $h \in [\ell]$ such that $|(\X_h, \X_h^1,\ldots,\X_h^t)-(\U_m,\X_h^1,\ldots,\X_h^t)| \le \epsilon$,
\end{itemize}
then 
\begin{align*}
|(\ell,t)\text{-}\nipm((\X,\Y), (\ell,t)\text{-}\nipm(\X^1,\Y^1),\ldots, (\ell,t)\text{-}\nipm(\X^t,\Y^t),\Y,\Y^1,\ldots,\Y^t\\ -\U_{m_1},  (\ell,t)\text{-}\nipm(\X^1,\Y^1),\ldots, (\ell,t)\text{-}\nipm(\X^t,\Y^t),\Y,\Y^1,\ldots,\Y^t| \le c_{\ref{thm:t-nipm}}'\ell \epsilon.
\end{align*}
\end{thm}
\begin{proof}  
We instantiate the $\ell$-look-ahead extractor described in Section $\ref{basic_merger}$ with the following strong-seeded extractors: Let $\Ext_1: \zo^m \times \zo^{d_1} \rightarrow  \zo^{d_1}$, $\Ext_2: \zo^d \times \zo^{d_1} \rightarrow \zo^{d_1}$ and $\Ext_3:\{ 0,1\}^{m} \times \{ 0,1\}^{d_1} \rightarrow \{ 0,1\}^{m_1}$ be explicit strong-seeded from Theorem $\ref{guv}$ designed to extract from min-entropy $k_1=m/2,k_2=d/2, k_3= m-c_{\ref{thm:t-nipm}}(t+1)\log(m/\epsilon)$ respectively with error $\epsilon$. Thus $d_1=c_{\ref{guv}}\log(m/\epsilon)$. 

For each $i\in[\ell-1]$, let $\Ext_{q,i}=\Ext_1$, $\Ext_{q,\ell}=\Ext_3$ and $\Ext_{w}=\Ext_2$.

Define $$t\text{-}\nipm((\X_1,\ldots,\X_{\ell}),\Y)=\ell\text{-}\laExt((\X_1,\ldots,\X_{\ell}),\Y).$$ For any random variable $\V=f((\X_1,\ldots,\X_t),\Y)$ (where $f$ is an arbitrary deterministic function), let $\V^{i}=f((\X_1^i,\ldots,\X_{\ell}^i),\Y^i)$.  The proof of correctness of the construction is similar in structure to Theorem $\ref{thm:nipm}$, but requires more care to handle $t$ adversaries.

We think of each $\X_i$ being uniform, and add back an error $\epsilon_1 \ell$ in the end.

We begin by proving the following claim.
\begin{claim}\label{cl:t-indep1}
For any $j\in [h-1]$, conditioned on the r.v's $\{\s_i: i \in [j-1] \}, \{ \s_i^{g}: i \in [j-1], g \in [t]\}, \{ \rr_i : i \in [j-1]\}, \{ \rr_i^g : i \in [j-1],  g \in [t]\}$ the following hold:
\begin{itemize}
\item $\s_{j}$ is $2(j-1)\epsilon$-close to $\U_d$,
\item $\s_{j},\{ \s_{j}^g: g \in [t]\}$ are deterministic functions of $\X, \{ \X_j^g: g \in [t] \}$,
\item for each $i \in [\ell]$, $\X_i$ has average conditional min-entropy at least $m-(t+1)(j-1)d_1-\log(1/\epsilon)$,
\item $\Y$ has average conditional min-entropy at least $k_1-2(t+1)(j-1)d_1-\log(1/\epsilon)$,
\item $\{ \X,\X^1,\ldots,\X^t\}$ is independent of $\{ \Y,\Y^1,\ldots,\Y^t\}$.
\end{itemize}
Further, conditioned on the r.v's $\{\s_i: i \in [j] \}, \{ \s_i^g: i \in [j], g \in [t]\}, \{ \rr_i : i \in [j-1]\}, \{ \rr_i^g : i \in [j-1], g \in [t]\}$ the following hold:
\begin{itemize}
\item $\rr_{j}$ is $(2j-1)\epsilon$-close to $\U_d$,
\item $\rr_{j},\{\rr_{j}^g: g\in [t] \}$ are deterministic functions of $\Y,\{\Y^g: g \in [t]\}$,
\item for any $i \in [\ell]$, $\X_i$ has average conditional min-entropy at least $m-(t+1)jd_1-\log(1/\epsilon)$,
\item $\Y$ has average conditional min-entropy at least $k_1-(t+1)(j-1)d_1-\log(1/\epsilon)$,
\item $\{ \X,\X^1,\ldots,\X^t\}$ is independent of $\{ \Y,\Y^1,\ldots,\Y^t\}$.
\end{itemize}
\end{claim}
\begin{proof} 
 In the course of the proof, we always maintain the property that the r.v's being fixed are either a deterministic function of $\{ \X,\X^1,\ldots,\X^t\}$ or $\{ \Y,\Y^1,\ldots,\Y^t\}$, and thus ensure $\{ \X,\X^1,\ldots,\X^t\}$ is independent of $\{ \Y,\Y^1,\ldots,\Y^t\}$. 
 
 We prove the claim by induction on $j$. The base case when $j=1$ is direct. Thus suppose $j>1$. 

Fix the r.v's $\{\s_i: i \in [j-1] \}, \{ \s_i^g: i \in [j-1], g \in [t]\}, \{ \rr_i : i \in [j-2]\}, \{ \rr_i^g : i \in [j-2], g \in [t]\}$. Using inductive hypothesis, it follows that 
\begin{itemize}
\item $\rr_{j-1}$ is $(2j-3)\epsilon$-close to $\U_{d_1}$,
\item $\rr_{j-1},\rr_{j-1}'$ are deterministic functions of $\{\Y,\Y' \}$,
\item for any $i \in [\ell]$, $\X_i$ has average conditional min-entropy at least $m-(j-1)d_1-\log(1/\epsilon)$,
\item $\Y$ has average conditional min-entropy at least $k_1-(j-2)d_1-\log(1/\epsilon)$,
\item $\{ \X,\X^1,\ldots,\X^t\}$ is independent of $\{ \Y,\Y^1,\ldots,\Y^t\}$.
\end{itemize}

Now since $\s_{j}=\Ext_1(\X_j,\rr_{j-1})$, it follows that $\s_{j}$ is $2(j-1)\epsilon$-close to $\U_{d_1}$ on average conditioned on $\rr_{j-1}$. We thus fix $\rr_{j-1}$, and $\s_j$ is now a deterministic function of $\X$.  Next, we fix $\{ \rr_{j-1}^g: g \in [t]\}$ without affecting the distribution of $\s_{j}$. Thus $\s_j,\s_j'$ are now a deterministic function of $\X,\X^1,\ldots,\X^t$. It follows that after these fixings,  the average conditional  min-entropy of $\Y$ is at least $k_1-(j-2)(t+1)d_1-\log(1/\epsilon)-(t+1)d_1=k_1-(j-1)(t+1)d_1-\log(1/\epsilon)$.

Next, we have $\rr_{j}=\Ext_2(\Y,\s_{j})$, and thus fixing $\s_{j}$, it follows that $\rr_j$ is $(2j-1)\epsilon$-close to uniform on average. Further, since $\rr_j$ is now a deterministic function of $\Y$, we fix $\{\s_j^g: g \in [t]\}$. As a result of these fixings, each $\X_i$ loses conditional min-entropy at most $2(t+1)d_1$ on average.  This completes the inductive step, and hence the proof follows.
\end{proof}

\begin{claim}\label{cl:t-indep2}Conditioned on the r.v's $\{\s_i: i \in [h-1] \}, \{ \s_i^g: i \in [h], g \in [t]\}, \{ \rr_i : i \in [h-1]\}, \{ \rr_i^g : i \in [h], g \in [t]\}$ the following hold:
\begin{itemize}
\item $\s_{h}$ is $2(h-1)\epsilon$-close to $\U_d$,
\item $\s_{h}$ is a deterministic function of $\X_h$,
\item for each $i \in [t]$, $\X_i$ has average conditional min-entropy at least $m-(t+1)hd_1-\log(1/\epsilon)$,
\item $\Y$ has average conditional min-entropy at least $k_1-(t+1)hd_1-\log(1/\epsilon)$,
\item $\{ \X,\X^1,\ldots,\X^t\}$ is independent of $\{ \Y,\Y^1,\ldots,\Y^t\}$.
\end{itemize}
\end{claim}
\begin{proof}We fix the r.v's $\{\s_i: i \in [h-1] \}, \{ \s_i^g: i \in [h-1], g \in [t]\}, \{ \rr_i : i \in [h-2]\}, \{ \rr_i^g : i \in [h-2], g \in [t]\}$, and using Claim  $\ref{cl:t-indep1}$ the following hold:
\begin{itemize}
\item $\rr_{h-1}$ is $(2h-3)\epsilon$-close to $\U_d$,
\item $\rr_{h-1},\{\rr_{h-1}^g: g \in [t]\}$ are deterministic functions of $\Y,\{\Y^g: g \in [t]\}$,
\item for any $i \in [\ell]$, $\X_i$ has average conditional min-entropy at least $m-(t+1)(h-1)d_1-\log(1/\epsilon)$,
\item $\Y$ has average conditional min-entropy at least $k_1-(t+1)(h-2)d_1-\log(1/\epsilon)$,
\item $\{ \X,\X^1,\ldots,\X^t\}$ is independent of $\{ \Y,\Y^1,\ldots,\Y^t\}$.
\end{itemize}
Next we claim that $\X_h$ has average conditional min-entropy at least $m-(h-1)(t+1)d_1-\log(1/\epsilon)$ even after fixing $\{\X_h^g: g \in [t]\}$. Before  fixings any other r.v, we have $\X_h|\{\X_h^g: g \in [t] \}$ is $\epsilon$-close to uniform on average. Since while computing the average conditional min-entropy, the order of fixing does not matter,  we can as well think of first fixing of $\{\X_h^g: g \in [t]\}$ and then fixing the r.v's $\{\s_i: i \in [h-1] \}, \{ \s_i': i \in [h-1]\}, \{ \rr_i : i \in [h-2]\}, \{ \rr_i' : i \in [h-2]\}$. Thus, it follows  that the average conditional min-entropy of $\X_h$ is at least $m-(t+1)(h-1)d_1-\log(1/\epsilon)$. 

We now prove that even after fixing the r.v's $\{\X_h^g: g \in [t]\},\rr_{h-1},\{\rr_{h-1}^g : g\in [t]\}$,  the r.v $\s_h$ is $2(h-1)\epsilon$-close to uniform on average. Fix $\{\X_h^g: g \in [t]\}$ and by the above argument  $\X_h$ has average conditional min-entropy at least $m-(t+1)(h-1)d_1-\log(1/\epsilon)$. Since $\s_h=\Ext_1(\X_h,\rr_{h-1})$, it follows that $\s_h$ is $2(h-1)\epsilon$-close to uniform on average conditioned on $\rr_{h-1}$. We fix $\rr_{h-1}$, and thus $\s_h$ is now a deterministic function of $\X_h$. Note that $\s_{h}^g=\Ext_1(\X_h^g,\rr_{h-1}^g)$ is now a deterministic function of $\rr_h^g$ (and thus $\Y^g$). Thus, we can fix $\{\rr_h^g: g \in [t]\}$ (which also fixes $\{ \s_h^g: g \in [t]\}$) without affecting the distribution of $\s_h$. 

Observe that once the r.v's $\rr_{h-1},\{ \rr_{h-1}^g: g \in [t]\}$ are fixed, $\{\s_h^g: g \in [t]\}$ is a deterministic function of $\{\X_h^g: g \in [t]\}$. We fix $\{ \s_h^g: g \in [t]\}$ and do not fix $\{ \X_h^g: g \in [t]\}$, and note that $\s_h$ is still $2(h-1)\epsilon$-close to uniform. Further after these fixings, each $\X_i$ has average conditional min-entropy at least $m-(t+1)hd_1-\log(1/\epsilon)$, and $\Y$ has average conditional min-entropy at least $k_1-(t+1)hd_1-\log(1/\epsilon)$.
\end{proof}
Theorem $\ref{thm:t-nipm}$ follows directly from the following claim.
\begin{claim}\label{cl:t-indep3}  For any $j\in [h,\ell]$, conditioned on the r.v's $\{\s_i: i \in [j-1] \}, \{ \s_i^g: i \in [j], g \in [t]\}, \{ \rr_i : i \in [j-1]\}, \{ \rr_i^g : i \in [j], g \in [t]\}$ the following hold:
\begin{itemize}
\item $\s_{j}$ is $2(j-1)\epsilon$-close to $\U_d$,
\item $\s_{j}$ is a deterministic function of $\X_j$
\item for each $i \in [\ell]$, $\X_i$ has average conditional min-entropy at least $m-(t+1)(j+1)d_1-\log(1/\epsilon)$,
\item $\Y$ has average conditional min-entropy at least $k_1-(t+1)jd_1-\log(1/\epsilon)$,
\item $\{ \X,\X^1,\ldots,\X^t\}$ is independent of $\{ \Y,\Y^1,\ldots,\Y^t\}$.
\end{itemize}
Further, conditioned on the r.v's $\{\s_i: i \in [j] \}, \{ \s_i^g: i \in [j+1], g \in [t]\}, \{ \rr_i : i \in [j-1]\}, \{ \rr_i^g : i \in [j], g \in [t]\}$ the following hold:
\begin{itemize}
\item $\rr_{j}$ is $(2j-1)\epsilon$-close to $\U_d$,
\item $\rr_{j}$ is a  deterministic function of $\Y$,
\item for any $i \in [\ell]$, $\X_i$ has average conditional min-entropy at least $m-(t+1)(j+1)d_1-\log(1/\epsilon)$,
\item $\Y$ has average conditional min-entropy at least $k_1-(t+1)jd_1-\log(1/\epsilon)$,
\item $\{ \X,\X^1,\ldots,\X^t\}$ is independent of $\{ \Y,\Y^1,\ldots,\Y^t\}$.
\end{itemize}
\end{claim}
\begin{proof}
We prove this by induction on $j$. For the base case, when $j=h$, fix the r.v's $\{\s_i: i \in [h-1] \}, \{ \s_i^g: i \in [h], g \in [t]\}, \{ \rr_i : i \in [h-1]\}, \{ \rr_i^g : i \in [h], g \in [t]\}$. Using Claim $\ref{cl:t-indep2}$, it follows that 
\begin{itemize}
\item $\s_{h}$ is $2(h-1)\epsilon$-close to $\U_d$,
\item $\s_{h}$ is a deterministic function of $\X_h$,
\item for each $i \in [\ell]$, $\X_i$ has average conditional min-entropy at least $m-(t+1)hd_1-\log(1/\epsilon)$,
\item $\Y$ has average conditional min-entropy at least $d-(t+1)hd_1-\log(1/\epsilon)$,
\item $\{ \X,\X^1,\ldots,\X^t\}$ is independent of $\{ \Y,\Y^1,\ldots,\Y^t\}$.
\end{itemize}
Noting that $\rr_h=\Ext_2(\Y,\s_h)$, we fix $\s_h$ and $\rr_h$ is $2h\epsilon$-uniform on average after this fixing. We note that $\rr_h$ is now a deterministic function of $\Y$. Since $\{\rr_h^g: g\in [t]\}$ is fixed, $\{\s_{h+1}^g: g \in [t]\}$ is a deterministic function of $\{\X_{h+1}^g:g \in [t]\}$, and we fix it without affecting the distribution of $\rr_h$. The average conditional min-entropy of each $\X_i$ after these fixings is at least $m-(t+1)(h+1)d_1-\log(1/\epsilon)$. 

Now suppose $j>h$.  Fix the r.v's  $\{\s_i: i \in [j-1] \}, \{ \s_i^g: i \in [j], g \in [t]\}, \{ \rr_i : i \in [j-2]\}, \{ \rr_i^g : i \in [j-1], g \in [t]\}$. By inductive hypothesis, the following hold:
\begin{itemize}
\item $\rr_{j-1}$ is $(2j-3)\epsilon$-close to $\U_d$,
\item $\rr_{j-1}$ is a  deterministic function of $\Y$,
\item for any $i \in [\ell]$, $\X_i$ has average conditional min-entropy at least $m-(t+1)jd_1-\log(1/\epsilon)$,
\item $\Y$ has average conditional min-entropy at least $k_1-(t+1)(j-1)d_1-\log(1/\epsilon)$,
\item $\{ \X,\X^1,\ldots,\X^t\}$ is independent of $\{ \Y,\Y^1,\ldots,\Y^t\}$.
\end{itemize}
Using the fact that $\s_{j}=\Ext_1(\X_j,\rr_{j-1})$, we fix $\rr_{j-1}$ and $\s_j$ is $(2j-2)\epsilon$-close to uniform on average after this fixing. Further, $\s_j$ is a deterministic function of $\X_j$. Since $\{\s_j^g: g \in [t]\}$ is fixed, it follows that $\{\rr_{j}^g: g \in [t]\}$ is a deterministic function of $\Y$ and we fix it without affecting the distribution of $\s_j$. We note that after these fixings, $\Y$ has average conditional min-entropy at least $k_1-(t+1)jd_1-\log(1/\epsilon)$. 

Now, we fix $\s_j$ and it follows that $\rr_j$ is a deterministic function of $\Y$ and is $(2j-1)\epsilon$-close to uniform on average. Further, since $\{\rr_j^g: g \in [t]\}$ is fixed, it follows that $\{\s_{j+1}^g: g \in [t]\}$ is a deterministic function of $\X_{j+1}$ and we fix it without affecting the distribution of $\rr_j$. The average conditional min-entropy of each $\X_i$ after these fixings is at least $m-(t+1)jd_1-\log(1/\epsilon)$. 

This completes proof of  the inductive step, and the claim now follows.
\end{proof}
\end{proof}

\subsection{A Recursive Non-Malleable Independence Preserving Merger}\label{sec_nm_rec}
In this section, we show a recursive way of applying the $(\ell,t)$-$\nipm$ constructed in the previous section in order to achieve better trade-off between parameters. This object is crucial in obtaining our near optimal non-malleable extractor construction.

\textbf{Notation:} For an $a \times b$ matrix $\V$, and any $S \subseteq [a]$, let $\V_{S}$ denote the matrix obtained by restricting $\V$ to the rows indexed by $S$.

Our main result in this section is the following theorem.
\begin{thm}\label{rec_thm}For all integers $m,\ell,L,t>0$, any $\epsilon>0$, $r=\lceil\frac{\log L}{\log \ell} \rceil$ and any $d=(c_{\ref{thm:t-nipm}} \ell \log(m/\epsilon) +d')(t+2)^{r+1}$,   there exists an explicit function $(L,\ell,t)$-$\nipm:\zo^{mL}\times \zo^{d} \rightarrow \zo^{m'}$, $m'=(0.9/t)^r (m-c_{\ref{thm:t-nipm}}\ell (t+1)r \log(m/\epsilon))$,   such that if the following conditions hold:
\begin{itemize}
\item $\X,\X^1,\ldots,\X^{t}$ are r.v's, each supported on boolean $L\times m$ matrices s.t for any $i \in [L]$, $|\X_i - \U_m| \le \epsilon$,
\item $\{\Y,\Y^1,\ldots,\Y^t\}$ is independent of $\{ \X,\X^1,\ldots,\X^{t}\}$, s.t $\Y,\Y^1,\ldots,\Y^t$ are each supported on $\zo^{d}$ and $H_{\infty}(\Y) \ge d-d'$,
\item there exists an $h \in [\ell]$ such that $|(\X_h,  \X_h^1,\ldots,\X_h^t)-(\U_{m},\X_h^1,\ldots,\X_h^t)| \le \epsilon$,
\end{itemize}
then 
\begin{align*}
|(L,\ell,t)\text{-}\nipm((\X,\Y), (L,\ell,t)\text{-}\nipm(\X^1,\Y^1),\ldots, (L,\ell,t)\text{-}\nipm(\X^t,\Y^t),\Y,\Y^1,\ldots,\Y^t \\ -\U_{m_1},  (L,\ell,t)\text{-}\nipm(\X^1,\Y^1),\ldots, (L,\ell,t)\text{-}\nipm(\X^t,\Y^t),\Y,\Y^1,\ldots,\Y^t| \le 2c_{\ref{thm:t-nipm}}' L\epsilon.
\end{align*}
\end{thm}
\begin{proof} We set up parameters and ingredients required in our construction.
\begin{itemize}
\item For $i \in [r]$, let $L_i = \lceil \frac{L}{\ell^{i}} \rceil$.
\item Let $d_1=d'+\log(1/\epsilon)+c_{\ref{thm:t-nipm}} (t+1) \ell \log(m/\epsilon)$. For $i \in [r]$, let $d_i= (t+2)d_{i-1}$.
\item Let $m_0=m$. For $i \in [r]$, define $m_i=0.9^i(m-ic_{\ref{thm:t-nipm}}(t+1)\ell \log(m/\epsilon))$
\item For each $i \in [r]$, let $(\ell,t)$-$\nipm_i:\{ 0,1\}^{\ell m_i} \times \zo^{d_i} \rightarrow \zo^{m_{i+1}}$ be an instantiation of the function from Theorem $\ref{thm:t-nipm}$ with error parameter $\epsilon$.
\end{itemize}
\RestyleAlgo{boxruled}
\LinesNumbered
\begin{algorithm}[ht]\label{alg1}
  \caption{$(L,\ell,t)\text{-}\nipm(x,y)$\label{alg}  \vspace{0.1cm}\newline \textbf{Input:}  $x$ is a boolean $L \times m$ matrix, and $y$ is a bit string of length $d$. \newline \textbf{Output:} A bit string of length $m_r$. }
  Let $x[0]=x$.
  
\For{$i=1$ to $r$} { 
		Let $y[i] = \slice(y,d_i)$
		
 		Let $x[i]$ be a $L_i \times m_i$ matrix, whose $j$'th row $x[i]_j = (\ell,t)\text{-}\nipm_i(x[i-1]_{[(j-1)\ell+1,j\ell]}, y[i])$
 	}

Ouput $x[r]$.
\end{algorithm}
We prove the following claim from which it is direct that the function $(L,\ell,t)$-$\nipm$ computed by Algorithm $\ref{alg1}$ satisfies the conclusion of Theorem $\ref{rec_thm}$. Let $\epsilon_0=\epsilon$, and for $i \in [r]$, let $\epsilon_i =\ell \epsilon_{i-1} +  c_{\ref{thm:t-nipm}}'\ell \epsilon$.
\begin{claim}\label{rec_claim} For all $ i \in [r]$, conditioned on the r.v's $\{ \Y[j]: j \in [i]\}, \{\Y^{g}[j]:  j \in [i], g \in [t] \}$, the following hold:
\begin{itemize}
\item $\X[i],\X^1[i],\ldots,\X^t[i]$ are r.v's, each  supported on boolean $L_i \times m_i$ matrices s.t for any $j \in [L_i]$, $|\X[i]_j - \U_{m_i}| \le (c'_{\ref{thm:t-nipm}} \ell)^i \epsilon$,
\item $\{\Y,\Y^1,\ldots,\Y^t\}$ is independent of $\{ \X[i],\X[i]^1,\ldots,\X[i]^{t}\}$.
\item there exists an $h_i \in [L_i]$ such that $\X[i]_h| \{ \X[i]^1_h,\ldots,\X[i]^t_h\}$ is $\epsilon_i$-close to $\U_{m_i}$ on average,
\item  $\Y$ has average conditional min-entropy at least $d - d_{i+1}+c_{\ref{thm:t-nipm}}(t+1)\ell \log(m/\epsilon)$.
\end{itemize}
\end{claim}
\begin{proof} We prove this claim by an induction on $i$. The base case, when $i=0$, is direct. Thus suppose $i \ge 1$. Fix the $r.v$'s $\{ \Y[j]: j \in [i-1]\}, \{\Y^{g}[j]:  j \in [i-1], g \in [t] \}$. By inductive hypothesis, it follows that 
\begin{itemize}
\item $\X[i-1],\X^1[i-1],\ldots,\X^t[i-1]$ are r.v's each  supported on boolean $L_{i-1} \times m_{i-1}$ matrices s.t for any $j \in [L_{i-1}]$, $|\X[i-1]_j - \U_{m-1}| \le (c_{\ref{thm:t-nipm}}'\ell)^{i-1} \epsilon$,
\item $\{\Y[i-1],\Y^1[i-1],\ldots,\Y^t[i-1]\}$ is independent of $\{ \X[i-1],\X[i-1]^1,\ldots,\X[i-1]^{t}\}$.
\item $h_i \in [L_i]$  such that $\X[i-1]_h| \{ \X[i-1]^1_h,\ldots,\X[i-1]_h^t\}$ is $\epsilon_{i-1}$-close to $\U_{m_{i-1}}$ on average,
\item  $\Y$ has average conditional min-entropy at least $d - d_{i}+c_{\ref{thm:t-nipm}}(t+1)\ell \log(m/\epsilon)$.
\end{itemize}
Thus the r.v $\Y[i]=\slice(\Y,d_i)$ has average conditional min-entropy at least  $c_{\ref{thm:t-nipm}}(t+1)\ell \log(n/\epsilon)$. Let $h_{i} \in [\ell(h_i-1)+1,\ell h_i]$, for some $h_i \in [L_i]$. By Claim $\ref{cl:t-indep1}$, it follows that conditioned on the r.v's $\Y[i], \{\Y^{g}[i]: g \in [t] \}$, for any $j \in [L_i]$,   $|\X[i]_j - \U_m| \le   \ell \epsilon_{i-1}+ c_{\ref{thm:t-nipm}}'\ell \epsilon=\epsilon_i$.  

Further, using Theorem $\ref{thm:t-nipm}$,  conditioned on $\Y[i], \{\Y^{g}[i]: g \in [t] \}, \{ \X^g[i]_{h_i}: g \in [t]\}$, the r.v $\X[i]_{h_i}$ is  $\ell \epsilon_{i-1}+ c_{\ref{thm:t-nipm}}'\ell \epsilon$-close to uniform on average. 

Thus, we fix the r.v's $\Y[i], \{\Y^{g}[i]: g \in [t] \}$, and note that  $\Y$ still has average conditional min-entropy at least $d -d_{i}+c_{\ref{thm:t-nipm}}(t+1)\ell \log(m/\epsilon)- (t+1)d_i \ge  d - d_{i+1}+c_{\ref{thm:t-nipm}}(t+1)\ell \log(m/\epsilon)$. This completes the proof of the inductive step, and the theorem follows.
\end{proof}
\end{proof}

\subsection{An Independence Preserving Merger Using a Weak Source}\label{sec_indep_weak}

In this section, we show a way of using the $(L,\ell,t)$-$\nipm$ constructed in the Section $\ref{sec_nm_rec}$ to merge the r.v's $\X,\X^1,\ldots,\X^{t}$, each supported on boolean $L\times m$ matrices, with the guarantee that there is some $h \in [L]$ s.t $\X_h$ is uniform on average conditioned on $\{ \X^g_h: g \in [t]\}$ using an independent $(n,k)$-source $\Y$ (instead of a seed as in the previous section). We note that our construction provides a direct improvement in terms of parameters over \cite{Coh16b}, and further uses just $1$ independent  source. In Section $\ref{multi_source}$, we use this new merger to improve upon the results on multi-source extractors obtained in \cite{Coh16b}.

We use the following notation, as introduced before.

\textbf{Notation:} For an $a \times b$ matrix $\V$, and any $S \subseteq [a]$, let $\V_{S}$ denote the matrix obtained by restricting $\V$ to the rows indexed by $S$.

Our main result in this section is the following theorem. We reuse the constants $c_{\ref{thm:t-nipm}}, c'_{\ref{thm:t-nipm}}$ from Theorem ${\ref{thm:t-nipm}}$.
\begin{thm}\label{rec_thm_weak}For all integers $m,\ell,L,t>0$, any $\epsilon>0$, $r=\lceil\frac{\log L}{\log \ell} \rceil$ and any $k\ge 2c_{\ref{thm:t-nipm}} \ell \log(m/\epsilon) (t+2)^{r+2}$,   there exists an explicit function $(L,\ell,t)$-$\ipm:\zo^{mL}\times \zo^{n} \rightarrow \zo^{m''}$, $m''=(0.9/t)^{r+1}(m-c_{\ref{thm:t-nipm}}\ell (t+1)r \log(m/\epsilon)- c_{\ref{guv}}(t+2)\log(n/\epsilon)))$,   such that if the following conditions hold:

\begin{itemize}
\item $\X,\X^1,\ldots,\X^{t}$ are r.v's, each supported on boolean $L\times m$ matrices s.t for any $i \in [L]$, $|\X_i - \U_m| \le \epsilon$,
\item $\Y$ is an $(n,k)$-source, independent of $\{ \X,\X^1,\ldots,\X^{t}\}$.
\item there exists an $h \in [\ell]$ such that $|(\X_h,  \X_h^1,\ldots,\X_h^t)-(\U_{m},\X_h^1,\ldots,\X_h^t)| \le \epsilon$,
\end{itemize}
then 
\begin{align*}
|(L,\ell,t)\text{-}\ipm(\X,\Y), (L,\ell,t)\text{-}\ipm(\X^1,\Y),\ldots, (L,\ell,t)\text{-}\nipm(\X^t,\Y)\\ -\U_{m''},  (L,\ell,t)\text{-}\ipm(\X^1,\Y),\ldots, (L,\ell,t)\text{-}\ipm(\X^t,\Y)| \le 3 c_{\ref{thm:t-nipm}}' L\epsilon.
\end{align*}
\end{thm}
\begin{proof} We set up parameters and ingredients required in our construction.
\begin{itemize}
\item Let $d=0.8k, d'=c_{\ref{guv}}\log(m/\epsilon), d_1=c_{\ref{guv}}\log(n/\epsilon)$.
\item Let $\Ext_1:\zo^{n} \times \zo^{d_1} \rightarrow \zo^d$ be a $(k,\epsilon)$-strong-seeded extractor from Theorem $\ref{guv}$.
\item Let $\Ext_2:\zo^{m} \times \zo^{d'} \rightarrow \zo^{m'}$, $m'=0.9(m-c_{\ref{guv}}(t+1)\log(n/\epsilon))$, be a $(m-c_{\ref{guv}}(t+1)\log(n/\epsilon),\epsilon)$-strong-seeded extractor from Theorem $\ref{guv}$.
\item Let $(L,\ell,t)$-$\nipm:\{ 0,1\}^{L m'} \times \zo^{d} \rightarrow \zo^{m''}$ be the function from Theorem $\ref{rec_thm}$ with error parameter $\epsilon$.
\end{itemize}
\RestyleAlgo{boxruled}
\LinesNumbered
\begin{algorithm}[ht]\label{alg2}
  \caption{$(L,\ell,t)\text{-}\ipm(x,y)$\label{alg}  \vspace{0.1cm}\newline \textbf{Input:}  $x$ is a boolean $L \times m$ matrix, and $y$ is a bit string of length $n$. \newline \textbf{Output:} A bit string of length $m''$. }
  Let $w=\slice(x_1,d_1)$
  
  Let $z=\Ext_1(y,w)$.
  
  Let $v=\slice(z,d')$.
  
  Let $\overline{v}$ be a $L \times m'$-matrix, whose $i$'th row is given by $\overline{v}_i=\Ext_2(x_i,v)$.
  
  Output $\overline{z}=(L,\ell,t)\text{-}\nipm(\overline{v},z)$.
\end{algorithm}
We begin by proving the following claim.
\begin{claim}\label{weak_c_1}Conditioned on $\W, \{\W^g: g \in [t]\}$, the following hold:
\begin{itemize}
\item $\Z$ is $\epsilon$-close to $\U_d$,
\item $\Z,\{\Z^g: g \in [t] \}$ is independent of $\X, \{ \X^g:g \in [t] \}$,
\item For each $i\in [L]$, $\X_i$ has average conditional min-entropy at least $m-(t+2)\log(n/\epsilon)$,
\item  $\X_h| \{\X_h^g: g \in [t] \}$ has average conditional min-entropy at least $m-(t+2)d_1\log(n/\epsilon)$.
\end{itemize}
\end{claim}
\begin{proof} Since $\Ext_1$ is a strong extractor, we can fix $\W$, and $\Z$ is $\epsilon$-close to $\U_d$ on average. Further, $\Z$ is now a deterministic function of $\X_1$. Thus, we can fix $\{\W^1,\ldots,\W^t\}$, without affecting the distribution of $\Z$. Since $\W^i$ is on $d_1$ bits, and without any prior conditioning since $\X| \{\X_h^g: g \in [t] \}$ is $\epsilon$-close to uniform on average, it follows that  conditioned on $\{\X_h^g: g \in [t] \},\W,\{ \W^g: g \in [t]\}$, the r.v $\X_h$ has average conditional min-entropy $m-(t+1)d_1\log(n/\epsilon)-\log(1/\epsilon)$.
\end{proof}

\begin{claim}\label{weak_c_2}Conditioned on $\W, \{\W^g: g \in [t]\}, \V,\{ \V^g:g \in [t]\}$, the following hold:
\begin{itemize}
\item $\{\Z,\Z^1,\ldots,\Z^t\}$ is independent of $\{\X,  \X^1,\ldots,\X^t\}$,
\item $\{ \overline{\V},\overline{\V}^1,\ldots,\overline{\V}^t\}$ is a deterministic function of $\{\X,\X^1,\ldots,\X^t\}$,
\item For each $i\in [L]$, $\overline{\V}_i$ is $2\epsilon$-close to uniform,
\item  $\overline{\V}_h| \{\overline{\V}_h^g: g \in [t] \}$ is $2\epsilon$-close to uniform on average.
\item  $\Z$ has average conditional min-entropy at least $d-(t+2)\log(m/\epsilon)$.
\end{itemize}
\end{claim}
\begin{proof} Fix $\W,\{ \W^{g}:g \in [t]\}$. Thus, by Claim $\ref{weak_c_1}$, we have 
\begin{itemize}
\item $\Z$ is $\epsilon$-close to $\U_d$,
\item $\Z,\{\Z^g: g \in [t] \}$ is independent of $\X, \{ \X^g:g \in [t] \}$,
\item For each $i\in [L]$, $\X_i$ has average conditional min-entropy at least $ m-(t+2)\log(n/\epsilon)$,
\item  $\X_h| \{\X_h^g: g \in [t] \}$ has average conditional min-entropy at least $m-(t+2)\log(n/\epsilon)$.
\end{itemize}

Since each $\X_i$ has  average conditional min-entropy at least $m-(t+2)\log(n/\epsilon)$, it follows that each $\overline{\V}_i$ is $2\epsilon$-close to uniform and  $\Ext_2$ is a strong extractor, it follows that $\overline{\V}_i$ is $2\epsilon$-close to $\U_d$ on average even conditioned on $\{\V,\V^1,\ldots,\V^t\}$.  After this fixing, $\Z$ has average conditional min-entropy at least $d-(t+2)\log(n/\epsilon)$. 

We now prove that $\overline{\V}_h| \{\overline{\V}_h^g: g \in [t] \}$ is $2\epsilon$-close to uniform on average. First, we fix the r.v's $\W,\{ \W^{g}:g \in [t]\}$ (at this point no other r.v's are fixed). As before, we have 
 $\X_h| \{\X_h^g: g \in [t] \}$ has average conditional min-entropy $k_x \ge m-(t+2)\log(n/\epsilon)$. Thus, we fix $\{\X_h^g: g \in [t] \}$. Now since $\Ext_2$ is a strong extractor, $\overline{\V}_h$ is uniform on average even conditioned on $\V$. We fix $\V$, and thus $\overline{\V}_h$ is a deterministic function of $\X_h$. Further, $\{\overline{\V}_h^g:g\in [t] \}$ is a deterministic function of $\{\V^g: g\in [t]\}$, and hence a deterministic function of $\Z,\{ \Z^g: g \in [t]\}$. Thus, we can fix $\{\overline{\V}_h^g:g\in [t] \}$ without affecting the distribution of $\overline{\V}_h$. This completes the proof of our claim.
\end{proof}

The correctness of the function $\ipm$ is direct from the next claim.
\begin{claim}\label{weak_c_3} Conditioned on $\{ \overline{\Z}^g: g \in [t]\}\}$, the r.v $\overline{\Z}$ is $3L\epsilon$-close to uniform on average.
\end{claim}
\begin{proof}Fix the r.v's $\W, \{\W^g: g \in [t]\}, \V,\{ \V^g:g \in [y]\}$. We observe that the following hold:
\begin{itemize}
\item $\Z,\{\Z^g: g \in [t] \}$ is independent of $\Y, \{ \Y^g:g \in [t] \}$,
\item For each $i\in [L]$, $\overline{\V}_i$ is $2\epsilon$-close to uniform,
\item  $\overline{\V}_h| \{\overline{\V}_h^g: g \in [t] \}$ is $2\epsilon$-close to uniform on average.
\item  $\Z$ has average conditional min-entropy at least $d-(t+2)\log(m/\epsilon)$.
\end{itemize}
The claim is now direct from Theorem $\ref{rec_thm}$ by observing that by our choice of parameters, the following hold:
\begin{itemize}
\item $d\ge (c_{\ref{thm:t-nipm}} \ell \log(m/\epsilon)+ d'' ) (t+2)^{r+1}$,  where $d''=(t+2)\log(m/\epsilon)$,
\item $\Z$ has average conditional min-entropy at least $d-d''$,
\item $m'' \le (0.9/t)^{r}(m'-c_{\ref{thm:t-nipm}}\ell (t+1)r \log(m/\epsilon))$.
\end{itemize}
This completes the proof of the claim, and hence Theorem $\ref{rec_thm_weak}$ follows.
\end{proof}
\end{proof}

\section{Explicit Almost-Optimal Non-Malleable Extractor}\label{sec:opt_nm}
We present an explicit construction of a non-malleable extractor with min-entropy requirement $k=(\log(n/\epsilon))^{1+o(1)}$ and seed-length $d=(\log(n/\epsilon))^{1+o(1)}$.  We also show a way of setting parameters that allows for $O(\log n)$ seed-length for large enough error. The following are the main results of this section.
\begin{thm}\label{main_nm_1} There exist a constant $ C_{\ref{main_nm_1}}>0$ s.t for all $n, k \in \mathbb{N}$ and any  $\epsilon>0$, with $k \ge \log(n/\epsilon)2^{C_{\ref{main_nm_1}}\sqrt{\log \log(n/\epsilon)}}$, there exists an explicit $(k,\epsilon)$-non-malleable extractor $\nmExt:\zo^n \times \zo^d \rightarrow \zo^{m}$, where $d= \log(n/\epsilon)2^{C_{\ref{main_nm_1}}\sqrt{\log \log(n/\epsilon)}}$ and $m= k/2^{\sqrt{\log \log (n/\epsilon)}}$.
\end{thm}

\begin{thm}\label{main_nm_2}There exist a constant $C_{\ref{main_nm_2}}>0$ s.t for constant $\beta>0$ and all $n, k \in \mathbb{N}$ and any  $\epsilon>2^{-\log^{1-\beta}(n)}$, with $k \ge C_{\ref{main_nm_2}}\log n$, there exists an explicit $(k,\epsilon)$-non-malleable extractor $\nmExt:\zo^n \times \zo^d \rightarrow \zo^{m}$, where $d= O(\log n)$ and $m=\Omega(\log(1/\epsilon))$.
\end{thm}

We derive both the above theorems from the following theorem.
\begin{thm}\label{nmext} There exist constants $\delta_{\ref{nmext}}, C_{\ref{nmext}}>0$ s.t for all $n, k \in \mathbb{N}$ and any error parameter $\epsilon_1>0$, with $k \ge \log(k/\epsilon_1)2^{C_{\ref{nmext}}\sqrt{\log \log(n/\epsilon_1)}}+C_{\ref{nmext}}\log(n/\epsilon_1)$, there exists an explicit $(k,\epsilon')$-non-malleable extractor $\nmExt:\zo^n \times \zo^d \rightarrow \zo^{m}$, where $d= \log(k/\epsilon)2^{C_{\ref{nmext}}\sqrt{\log \log(n/\epsilon_1)}}+C_{\ref{nmext}}\log(n/\epsilon_1), m=\delta_{\ref{nmext}} k/2^{\sqrt{\log \log (n/\epsilon_1)}}$ and $\epsilon'=C_{\ref{nmext}}\epsilon_1\log(n/\epsilon_1)$.
\end{thm}
We first show how to derive Theorem $\ref{main_nm_1}$ and Theorem $\ref{main_nm_2}$ from Theorem $\ref{nmext}$.
\begin{proof}[Proof of Theorem $\ref{main_nm_1}$] Let $\nmExt:\zo^n \times \zo^d \rightarrow \zo^m$ be the function from Theorem $\ref{nmext}$ set to extract from min-entropy $k$, where we set the parameter $\epsilon_1= \epsilon/2C_{\ref{nmext}}n$ . It follows that the error of $\nmExt$ is $$C_{\ref{nmext}}\epsilon_1\log(n/\epsilon_1) =\frac{\epsilon}{2n}(\log n + \log(2C_{\ref{nmext}}n)+\log(1/\epsilon)) <\epsilon.$$  Further note that for this setting of $\epsilon_1$, the min-entropy required and seed length are $\log(n/\epsilon)2^{C_{\ref{main_nm_1}}\sqrt{\log \log(n/\epsilon)}}+C_{\ref{main_nm_1}}\log(n/\epsilon)$, for some constant $C_{\ref{main_nm_1}}$.
\end{proof}
\begin{proof}[Proof of Theorem $\ref{main_nm_2}$]Let $\nmExt:\zo^n \times \zo^d \rightarrow \zo^m$ be the function from Theorem $\ref{nmext}$ set to extract from min-entropy  $2C_{\ref{nmext}}\log(n/\epsilon_1)$, where we set the parameter $\epsilon_1=\epsilon/2C_{\ref{nmext}} \log n$. Thus, the error of $\nmExt$ is  $$\epsilon_1\log(n/\epsilon_1) \le \frac{\epsilon}{2 \log n}(\log n+ \log(1/\epsilon)+\log(2C_{\ref{nmext}} \log n)) <\epsilon.$$ For this setting of parameters, we note that the seed-length required by  $\nmExt$ is bounded by $\log((\log^2 n)/\epsilon)2^{C_{\ref{nmext}}\sqrt{\log \log(n\log n/\epsilon)}}+C_{\ref{nmext}}\log(n\log n/\epsilon)=O(\log n)$.
\end{proof}

We spend the rest of the section proving Theorem $\ref{nmext}$. We recall some explicit constructions from previous work.

The following flip-flop function  was constructed by Cohen \cite{Coh15} using alternating extraction. Subsequently, Chattopadhyay, Goyal and Li \cite{CGL15}, used this  in constructing non-malleable extractors. Informally, the flip-flop function uses an independent source $\X$ to break the correlation between two r.v's $\Y$ and $\Y'$, given an advice bit. We now describe this more formally.
\begin{thm}[\cite{Coh15,CGL15}]\label{flip} There exist constants $C_{\ref{flip}},\delta_{\ref{flip}}>0$ such that for all $n>0$ and any $\epsilon>0$, there exists an explicit function $\flip:\zo^n \times \zo^d \rightarrow \zo^m$, $m=\delta_{\ref{flip}} k$,   satisfying the following: Let $\X$ be an $(n,k)$-source, and $\Y$ be an independent weak seed on $d$ bits with entropy $d-\la$, $\la < d/2$. Let $\Y'$ be a r.v on $d$ bits independent of $\X$,  and let $b,b'$ be bits s.t. $b \neq b'$. If $k,d \ge C_{\ref{flip}}\log(n/\epsilon)$,  then 
$$|\flip(\X,\Y,b),\flip(\X,\Y',b'),\Y,\Y' - \U_m,\flip(\X,\Y',b'),\Y,\Y'| \le \epsilon.$$
\end{thm}
We  now recall an explicit function $\adv$ from \cite{CGL15}. Informally, $\adv$ takes as input a source $\X$ and a seed $\Y$ and produces a short string such that for any r.v $\Y' \neq \Y$, $\adv(\X,\Y) \neq \adv(\X,\Y)$. We record this property more formally.
\begin{thm}[\cite{CGL15}]\label{adv_gen} There exists a constant $c_{\ref{adv_gen}},C_{\ref{adv_gen}}>0$ such that  for all $n>0$ and any $\epsilon>0$, there exists an explicit function $\adv:\zo^n \times \zo^d \rightarrow \zo^{L}$, $L=c_{\ref{adv_gen}} \log (n/\epsilon)$ satisfying the following: Let $\X$ be an $(n,k)$-source, and $\Y$ be an independent uniform seed on $d$ bits. Let $\Y'$ be a r.v on $d$ bits independent of $\X$, s.t $\Y' \neq \Y$. If $k,d \ge C_{\ref{adv_gen}} \log(n/\epsilon)$, then 
\begin{itemize}
\item with probability at least $1-\epsilon$, $\adv(\X,\Y) \neq \adv(\X,\Y')$,
\item  there exists a function $f$ such that conditioned on $\adv(\X,\Y), \adv(\X,\Y'), f(\X)$,
\begin{itemize}
\item $\X$ remains independent of $\Y,\Y'$,
\item $\X$ has average conditional min-entropy at least $k-C_{\ref{adv_gen}}\log(n/\epsilon)$,
\item $\Y$ has average conditional min-entropy at least $d-C_{\ref{adv_gen}}\log(n/\epsilon)$
\end{itemize}
\end{itemize}
\end{thm}

We are now ready to prove Theorem $\ref{nmext}$.

\begin{proof}[Proof of Theorem $\ref{nmext}$]We set up parameters and ingredients required in our construction.
\begin{itemize}
\item Let $\adv:\zo^n \times \zo^d \rightarrow \zo^{L}$, $L=c_{\ref{adv_gen}}\log(n/\epsilon_1)$, be the function from Theorem $\ref{adv_gen}$ with error parameter $\epsilon_1$.
\item Let $d_1= (C_{\ref{adv_gen}}+ C_{\ref{flip}}+1)\log(n/\epsilon_1)$.
\item Let $\flip:\zo^n \times \zo^{d_1} \rightarrow \zo^{m'}$, $m'=\delta_{\ref{flip}}k$, be the function from Theorem $\ref{flip}$ with error parameter $\epsilon_1$.
\item $d_2=c_{\ref{guv}}\log(d/\epsilon_1), d_3=c_{\ref{guv}}\log(m'/\epsilon_1)$.
\item Let $\Ext_1:\zo^{n} \times \zo^{d_2} \rightarrow \zo^{d'}$, $d'=0.9d-2d_1-C_{\ref{adv_gen}}\log(n/\epsilon_1)$ be a $(d-2d_1-C_{\ref{adv_gen}}\log(n/\epsilon_1),\epsilon_1)$-strong-seeded extractor from Theorem $\ref{guv}$.
\item Let $\Ext_2:\zo^{m'} \times \zo^{d_3} \rightarrow \zo^{m''}$, $m''=0.9m'-2d_2$, be a $(m'-2d_2-\log(1/\epsilon_1),\epsilon_1)$-strong-seeded extractor from Theorem $\ref{guv}$.
\item Let $\ell = 2^{\sqrt{\log L}}$.
\item Let $(L,\ell,1)$-$\nipm:\zo^{Lm''} \times \zo^{d'} \rightarrow \zo^{m}$ be the function from Theorem $\ref{rec_thm}$, $m=0.9^rm'-2c_{\ref{thm:t-nipm}}\ell(t+1)r\log(m/\epsilon_1)$ with error parameter $\epsilon_1$.
\end{itemize}
\RestyleAlgo{boxruled}
\LinesNumbered
\begin{algorithm}[ht]\label{alg3}
  \caption{$\nmExt(x,y)$\label{alg}  \vspace{0.1cm}\newline \textbf{Input:}  $x,y$ are  bit string of length $n,d$ respectively. \newline \textbf{Output:} A bit string of length $m$. }
    Let $w=\adv(x,y)$.
  
    Let $y=y_1 \circ y_2$, where $y_1=\slice(y,d_1)$.

  Let $v$ be a $L \times m'$ matrix, whose $i$'th row $v_i= \flip(x,y_1,w_i)$ ($w_i$ is the $i$'th bit of the string $w$).
  
  Let $\overline{v_1}=\slice(v_1,d_2)$
  
  Let $\overline{y}=\Ext_1(y,\overline{v_1})=\overline{y_1} \circ \overline{y_2}$, where $\overline{y_1}=\slice(\overline{y},d_3)$.
  
  Let $z$ be a $L \times m''$ matrix,  whose $i$'th row $z_i= \Ext_2(v_i,\overline{y_1})$

  Output $\overline{z}=(L,\ell,1)\text{-}\nipm(z,\overline{y})$.
\end{algorithm}
We prove in the following claims that the function $\nmExt$ constructed in Algorithm $\ref{alg2}$ satisfies the conclusion of Theorem $\ref{nmext}$. Let $\A$ be the adversarial function tampering the seed $\Y$, and let $\Y'=\A(\Y)$. Since $\A$ has no fixed points, it follows that $\Y \neq \Y'$.

\textbf{Notation:}  For any random variable $\mathbf{H}=g(\X,\Y)$ (where $g$ is an arbitrary deterministic function), let $\mathbf{H}^{\prime}=g(\X,\Y')$.
\begin{claim}\label{nm_c_1} With probability at least $1-\epsilon$, $\W \neq \W'$.
\end{claim}
\begin{proof}Follows directly from Theorem $\ref{adv_gen}$.
\end{proof}
Let  $f$ be the function guaranteed by Theorem $\ref{adv_gen}$.
\begin{claim}\label{nm_c_2}Conditoned on the r.v's $\W,\W',\Y_1,\Y_1',f(\X)$, the following hold:
\begin{itemize}
\item for each $i \in [L]$, $\V_i$ is $\epsilon_1$-close to uniform,
 \item there exists an $h \in [L]$ such that conditioned on $\V_h'$, the r.v $\V_h$ is $\epsilon_1$-close to uniform on average,
 \item $\{ \V,\V'\}$ is independent of $\{ \Y,\Y'\}$.
 \item $\Y$ has average conditional min-entropy at least  $d-C_{\ref{adv_gen}}\log(n/\epsilon_1)-2d_1$.
 \end{itemize}
\end{claim}
\begin{proof} Fix the r.v's $\W,\W',f(\X)$ such that $\W \neq \W'$. It follows from Theorem $\ref{adv_gen}$ that after this conditioning,
\begin{itemize}
\item $\X$ is independent of $\Y,\Y'$,
\item $\X$ has average conditional min-entropy at least $k-C_{\ref{adv_gen}}\log(n/\epsilon_1)$,
\item $\Y$ has average conditional min-entropy at least $d-C_{\ref{adv_gen}}\log(n/\epsilon_1)$
\end{itemize}
Thus $\Y_1=\slice(\Y,d_1)$ has average conditional min-entropy at least $2C_{\ref{flip}}\log(n/\epsilon_1)$. The claim now follows by applying Theorem $\ref{flip}$.
\end{proof}

\begin{claim}\label{nm_c_3} Conditioned on the r.v's $\W,\W',\overline{\V}_1,\overline{\V_1}', \Y_1,\Y_1' \overline{\Y_1}, \overline{\Y_1}',f(\X)$, the following hold:
\begin{itemize}
\item $\overline{\Y}$ has average conditional min-entropy at least $d'-2d_3-\log(1/\epsilon)$.
\item for each $i \in [L]$, $\Z_i$ is $3\epsilon_1$-close to uniform on average.
\item there exists $h \in [L]$ such that further conditioned on $\Z_i'$, $\Z_i$ is $3\epsilon_1$-close to uniform on average.
\item $\{ \overline{\Y},\overline{\Y}'\}$ is independent of $\{ \Z,\Z'\}$.
\end{itemize}
\end{claim}
\begin{proof}Fix the r.v's $\W,\W',\Y_1,\Y_1',f(\X)$. By Claim $\ref{nm_c_2}$, we have 
\begin{itemize}
\item for each $i \in [L]$, $\V_i$ is $\epsilon_1$-close to uniform,
 \item there exists an $h \in [L]$ such that conditioned on $\V_h'$, the r.v $\V_h$ is $\epsilon_1$-close to uniform on average,
 \item $\{ \V,\V'\}$ is independent of $\{ \Y,\Y'\}$.
 \item $\Y$ has average conditional min-entropy at least  $d-C_{\ref{adv_gen}}\log(n/\epsilon_1)-2d_1$.
 \end{itemize}
 Using the fact that $\Ext_1$ is a strong extractor, it follows that we can fix $\overline{\V_1}$, and $\overline{\Y}$ is $2\epsilon_1$-close to uniform on average. Further, $\overline{\Y}$ is a deterministic function of $\Y$. Thus, we fix $\overline{\V_1}'$ without affecting the distribution of $\overline{\Y}$. Now, using the fact that $\Ext_2$ is a strong extractor, we can fix $\overline{\Y_1}$, and we have  for each $i \in [L]$, $\Z_i$ is $3\epsilon_1$-close to uniform on average. Next we can fix $\overline{\Y_1}'$ without affecting $\V$.
 
We prove that conditioned on $\Z_i'$, the r.v $\Z_i$ is $3\epsilon_1$-close to uniform on average in the following way. For this argument, as above we fix all r.v's but do not yet fix $\overline{\Y_1},\overline{\Y_1}'$. Instead, we first fix $\V_h'$, and  $\V_h$ has average conditional min-entropy at least $m'-2d_2$. We now fix $\overline{\Y_1}$, and as before we have $\Z_h$ is $3\epsilon_1$-close. At this point, $\Z_h'$ is a deterministic function of $\overline{\Y_1}'$, and hence we can fix it without affecting the distribution of $\Z_h$. This completes the proof.
\end{proof}

\begin{claim}\label{nm_c_4}Conditioned on $\overline{\Z}'$, the r.v $\overline{\Z}$ is $O(\epsilon_1\log(n/\epsilon_1))$-close to uniform on average.
\end{claim}
\begin{proof}Fix the r.v's $\W,\W',\overline{\V}_1,\overline{\V_1}', \Y_1,\Y_1'\overline{\Y_1}, \overline{\Y_1}',f(\X)$. By Claim $\ref{nm_c_3}$, the following hold:
\begin{itemize}
\item $\overline{\Y}$ has average conditional min-entropy at least $d'-2d_3-\log(1/\epsilon_1)$.
\item for each $i \in [L]$, $\Z_i$ is $3\epsilon_1$-close to uniform on average.
\item there exists $h \in [L]$ such that further conditioned on $\Z_i'$, the r.v $\Z_i$ is $3\epsilon_1$-close to uniform on average.
\item $\{ \overline{\Y},\overline{\Y}'\}$ is independent of $\{ \Z,\Z'\}$.
\end{itemize}
Let  $d''=2d_3+\log(1/\epsilon_1), r= \lceil \frac{\log L}{\log \ell}\rceil =\lceil \sqrt{\log L}\rceil$. Thus $d''=O(\log(k/\epsilon_1)), r=O(\sqrt{\log \log (n/\epsilon_1)}),  \ell = 2^{O(\sqrt{\log \log(n/\epsilon_1) })}$. In order to use Theorem $\ref{rec_thm}$, we observe that for a large enough constant $C_{\ref{nmext}}$ the following hold:
\begin{itemize}
\item $\overline{\Y}$ has conditional min-entropy at least $d-d''$,
\item $d' \ge (c_{\ref{thm:t-nipm}}\ell \log(m''/\epsilon_1)+d'') 3^{r+1}$,
\item $m<(0.9)^r(m''-c_{\ref{thm:t-nipm}}\ell(t+1)r\log(m/\epsilon_1))$.
\end{itemize}
Thus the conditions of Theorem $\ref{rec_thm}$ are met, and hence it follows that conditioned on $\overline{\Z}'$, the r.v $\overline{\Z}$ is $2c'_{\ref{thm:t-nipm}}L\epsilon_1$-close to uniform on average. Recall that $L=O(\log(n/\epsilon_1))$, and hence the claim follows.
\end{proof}

\end{proof}
\section{Improved $t$-Non-Malleable Extractors and $2$-Source Extractors}\label{sec_t}
The framework to construct non-malleable extractors in Section $\ref{sec:opt_nm}$ can be generalized directly to construct non-malleable extractors that can handle multiple adversaries.
\begin{define}[$t$-Non-malleable Extractor]\label{def_tnmext} A function t-$\nmExt:\{0,1\}^n \times \{ 0,1\}^d \rightarrow \{ 0,1\}^m$ is a seeded $t$-non-malleable extractor for min-entropy $k$ and error $\epsilon$ if the following holds : If $X$ is a  source on  $\{0,1\}^n$ with min-entropy $k$ and $\A_1 : \{0,1\}^n \rightarrow \{0,1\}^n,\ldots,  \A_t : \{0,1\}^n \rightarrow \{0,1\}^n$ are  arbitrary tampering function with no fixed points, then
\begin{align*}
|t\text{-}\nmExt(\X,\U_d), t\text{-}\nmExt(\X,\A_1(\U_d)),\ldots,t\text{-}\nmExt(\X,\A_t(\U_d)),\U_d \\ -  \U_m \scirc t\text{-}\nmExt(\X,\A_1(\U_d)),\ldots,t\text{-}\nmExt(\X,\A_t(\U_d)),U_d  | <\epsilon 
\end{align*}
\end{define}

 In particular, Theorem $\ref{flip}$ and Theorem $\ref{adv_gen}$ both generalize to the case there are $t$ tampered variables, and further our $\nipm$ construction in Theorem $\ref{thm:t-nipm}$ handles $t$ adversaries. By using these versions of the components in the above construction,  the following theorem is easy to obtain. Since the proof is similar to the proof of Theorem $\ref{nmext}$, we omit the proof of the following theorem.
\begin{THM}\label{nmext_multi}There exists a constant $\delta>0$ such that for all $n, k, t, \ell \in \mathbb{N}$ and any  $\epsilon>0$, with $r=(\log \log (n/\epsilon))/(\log \ell)$, $k = \Omega(t^{2r}\ell \log(n/\epsilon))$, there exists an explicit $(t,k,\epsilon)$-non-malleable extractor $\nmExt:\zo^n \times \zo^d \rightarrow \zo^{m}$, where $d= O(t^{(1+\delta)r}\ell \log(n/\epsilon))$ and $m= (\delta k-\ell t r \log(n/\epsilon))/(2t)^{(\log L/\log \ell)}$.
\end{THM}

As discussed in the introduction, such non-malleable extractors were used in \cite{CZ15} to construct two-source extractors, with subsequent improvements in parameters \cite{Li:2source,Mek:resil}. Combining the framework of \cite{CZ15}, with the improved components from \cite{Li:2source,Mek:resil} and our new $t$-non-malleable extractor from Theorem $\ref{nmext_multi}$, the following results are easy to obtain by suitably optimizing parameters.

\begin{THM}\label{2ext1}There exists a constant $C>0$ such that for any $\delta>0$ and for all $n,k \in \mathbb{N}$ with $k \ge  C(\log n)^{2\sqrt{6(1+\delta)}+3}$ and any constant $\epsilon<\frac{1}{2}$ , there exists an efficient polynomial time computable $2$-source extractor min-entropy $k$ with error $\epsilon$ that outputs $1$ bit.
\end{THM}

\begin{THM}\label{2ext2}There exists a constant $C>0$ such that for any $\delta>0$ and for all $n,k \in \mathbb{N}$ with $k \ge C(\log(n))^{4\sqrt{5(1+\delta)}+5}$, there exists an efficient polynomial time computable $2$-source extractor min-entropy $k$ with error $n^{-\Omega(1)}$ and output length $\Omega(k)$.
\end{THM}

\input{newmerger.tex}

\section{Improved Multi-Source Extractors}\label{multi_source}
In this section, we construct extractors for a constant number of independent sources $\X_1,\ldots,\X_C$, each with min-entropy $\tilde{O}(\log n)$. In particular, this improves upon a recent result of Cohen and Schulman \cite{Coh16b}, where they constructed an extractor for $O(1/\delta)$ independent sources, with each having min-entropy $\log^{1+\delta}(n)$.

Our main result in this section is the following.
\begingroup
\def\theTHM{\ref{multi_thm}}
\begin{THM}[restated]
 There exists a constant $C>0$ s.t for all $n, k \in \mathbb{N}$ and any constant $\epsilon>0$, with $k \ge  2^{C\sqrt{\log \log(n)}}\log n$, there exists an explicit function $\Ext:(\zo^n)^{C} \rightarrow \zo$, such that $$|\Ext(\X_1,\ldots,\X_{C}) -\U_1 | \le \epsilon.$$
\end{THM}
\addtocounter{THM}{-1}
\endgroup

Our starting point is the following reduction from \cite{Coh16b}. Informally, a constant number of independent sources are used to transform into a sequence of matrices such that a large fraction of the matrices follow a certain $t$-wise independence property. For our purposes, we need to slightly modify this construction. The length of the rows (the parameter $m$ in the following theorem) in the work of \cite{Coh16b} can be set to $c\log(n/\epsilon)$, for any constant $c$. Using another additional source and extracting from it using each row as seed (using any optimal strong-seeded extractor), the length of each row can be made $\Omega(k)$. We state the theorem from \cite{Coh16b} with this modification.

\begin{thm}[\cite{Coh16b}]\label{sr_indep}There exists constants $\alpha>0$ and and $c_{\ref{sr_indep}}$ such that for all $n,t \in \mathbb{N}$, and for any $\epsilon,\delta>0$, there exists an polynomial time computable function $f:(\zo^n)^{C} \rightarrow (\zo^{Lm})^r$, where $C=7/\alpha, L=O(t \log n), r=n^{3/\alpha}, m=\Omega(k)$, such that the following hold: Let $\X_1,\ldots,\X_C$ be independent $(n,k)$ sources, $k=c_{\ref{sr_indep}}t \log(t) \log(n\log t/\epsilon)$. Then there exists a subset $S \subset [r]$, $|S| \ge r-r^{\frac{1}{2}-\alpha}$ and a sequence of $L \times m$ matrices $\Y^1,\ldots,\Y^r$ such that: 
\begin{itemize}
\item $f(\X_1,\ldots,\X_C)$ is $1/r$-close to $\Y^1,\ldots,\Y^r$,
\item for any $i \in [L]$ and $g\in S$, $\Y^g_i$ is $\epsilon$-close to $\U_m$,
\item for any $g \in S$, and any distinct $i_1,\ldots,i_t$ in $S \setminus \{ g\}$, there exists an $h \in [L]$ such that $\Y^g_h|\{\Y^j_h: j\in [r]\setminus\{ g\}\}$ is $\epsilon$-close to uniform.
\end{itemize}
 \end{thm}
 
  Now composing the above theorem with our independence preserving merger from Section $\ref{sec_indep_weak}$, we have the following result.
 \begin{thm}\label{good_bits}There exists a constant $\alpha>0$  such that for all $n,t \in \mathbb{N}$, and for any $\epsilon,\delta>0$, there exists an polynomial time computable function $\reduce:(\zo^n)^{C+1} \rightarrow \zo^{r}$, where $C=\frac{7}{\alpha}+1,  r=n^{3/\alpha}$, such that the following hold: Let $\X_1,\ldots,\X_C$ be independent $(n,k)$ sources, $k\ge  2^{\sqrt{\log t+ \log \log n}} \log(k/\epsilon) (t+2)^{O(\sqrt{\log t+\log \log n})}+c_{\ref{sr_indep}}t \log(t) \log(n\log t/\epsilon)$, and let $\Z=\reduce(\X_1,\ldots,\X_{C+1})$. Then there exists a subset $S \subset [r]$, $|S| \ge r- r^{\frac{1}{2}-\alpha}$ such that $\Z_{S}$ is $n^{-\Omega(1)}$-close to a $(t,\gamma_{\ref{good_bits}})$-wise independent distribution, where $\gamma_{\ref{good_bits}}=O(\epsilon t \log n)$.
 \end{thm}
 \begin{proof} Let $f:(\zo^n)^{C} \rightarrow (\zo^{Lm})^r$ be the function from Theorem $\ref{sr_indep}$ with $\epsilon_{\ref{sr_indep}}=\epsilon, m=\beta k$ for some constant $\beta>0$. Thus $L=O(t \log n)$. Let $(L,\ell,t)$-$\ipm:(\zo^{Lm})^t \times \zo$ be the function from Theorem $\ref{rec_thm_weak}$, with $\ell=2^{\sqrt{\log L}}=2^{O(\sqrt{\log t+\log \log n})}$ and error parameter $\epsilon_{\ref{rec_thm_weak}}=\epsilon$. Define $$\reduce(x_1,\ldots,x_{C+1}) = (L,\ell,t)\text{-}\ipm(f(x_1,\ldots,x_{C}),x_{C+1}).$$
 We note that $k >c_{\ref{sr_indep}}t \log(t) \log(n\log t/\epsilon)$. Thus, using Theorem $\ref{sr_indep}$, it follows that there exists a subset $S \subset [r]$, $|S| \ge r-r^{\frac{1}{2}-\alpha}$ and a sequence of $L \times m$ matrices $\Y^1,\ldots,\Y^r$ such that: 
\begin{itemize}
\item $f(\X_1,\ldots,\X_C)$ is $1/r$-close to $\Y^1,\ldots,\Y^r$,
\item for any $i \in [L]$ and $g\in S$, $\Y^g_i$ is $\epsilon$-close to $\U_m$,
\item for any $g \in S$, and any distinct $i_1,\ldots,i_t$ in $S \setminus \{ g\}$, there exists an $h \in [L]$ such that $\Y^g_h|\{\Y^j_h: j\in [r]\setminus\{ g\}\}$ is $\epsilon$-close to uniform.
\end{itemize}
 We now work with the sources $\Y^1,\ldots,\Y^r$, and add an error of $1/r$ in the end. The theorem is now direct using Theorem $\ref{rec_thm_weak}$ and observing that the following hold  by our setting of parameters:
 \begin{itemize}
 \item $k \ge 2c_{\ref{thm:t-nipm}} \ell \log(k/\epsilon) (t+2)^{\lceil \frac{\log L}{\log \ell}\rceil+1}$,
 \item $m=\beta k  \ge 2^{\sqrt{\log L}}(c_{\ref{thm:t-nipm}}\ell(t+1)r\log(m/\epsilon)+c_{\ref{guv}}(t+2)\log(n/\epsilon))$.
 \end{itemize}
 \end{proof}

 Our multi-source extractor in Theorem $\ref{multi_thm}$ is now easy to obtain using a result on the majority function. 
 
 \begin{thm}[\cite{DGJSV,Viola14,Coh16b}] Let $\Z$ be a source on $r$ bits such that there exists a subset $S \subset [r]$, $|S| \ge r- r^{\frac{1}{2}-\alpha}$ such that $\Z_S$ is $t$-wise independent. Then, $$\l|\Pr[\textnormal{Majority}(\Z)=1]-\frac{1}{2}\r| \le O\l(\frac{\log t}{t}+r^{-\alpha}\r).$$
 \end{thm}
 
 We also recall a result about almost $t$-wise independent distributions.
 \begin{thm}[\cite{AGM03}]\label{almost_t_wise} Let $\D$ be a $(t,\gamma)$-wise independent distribution on $\{ 0,1\}^n$. Then there exists a $t$-wise independent distribution that is $n^{t}\gamma$-close to  $\D$.
\end{thm}

Thus, we have the following corollary.  
  \begin{cor}\label{cor_majority}There exists a constant $c$ such that the following holds: Let $\Z$ be a source on $r$ bits such that there exists a subset $S \subset [r]$, $|S| \ge r- r^{\frac{1}{2}-\alpha}$ such that $\Z_S$ is $(t,\gamma)$-wise independent. Then, $$|\Pr[\textnormal{Majority}(\Z)=1]-\frac{1}{2}| \le c\l(\frac{\log t}{t}+r^{-\alpha}+ \gamma r^t\r).$$
 \end{cor}
 
 \begin{proof}[Proof of Theorem $\ref{multi_thm}$] Set $t$ to a large enough constant such that $\frac{c\log t}{t}<\epsilon/2$. Let $\alpha$ be the constant from Theorem $\ref{good_bits}$, $r=n^{3/\alpha}$ and $C=\frac{7}{\alpha}+1$. Let $\reduce$ be the function from Theorem $\ref{good_bits}$ with parameter $t_{\ref{good_bits}}=t$, $r_{\ref{good_bits}}=r$, and the error parameter $\epsilon_{\ref{good_bits}}$ set such that the parameter $\gamma_{\ref{good_bits}} \le \frac{1}{r^{t+1}}$. This can be ensured by setting $\epsilon = n^{-C'}$ for a large enough constant $C'$.
 
  Define $$\Ext(x_1,\ldots,x_C)=\textnormal{Majority}(f(x_1,\ldots,x_C)).$$

Let $\Z=f(\X_1,\ldots,\X_C)$. We  note that with this setting of parameters, there exists some constant $C''$ such that any $k \ge 2^{C''\sqrt{\log \log n}}\log(n)$ is sufficient for the conclusion of Theorem $\ref{good_bits}$ to hold. Thus, $\Z$ is a source on $r$ bits such that there exists a subset  $S \subset [r]$, $|S| \ge r-r^{\frac{1}{2}-\alpha}$ for which $\Z_{S}$ is $(t,\gamma)$-wise independent.  Theorem $\ref{multi_thm}$ is now direct from Corollary $\ref{cor_majority}$.
 \end{proof}


%% file: newmerger.tex
\section{A More Involved $\nipm$ and Non-Malleable Extractor}\label{sec:involved}
In this section we use our previous $\nipm$ to construct a more involved $\nipm$, which can be used to give explicit non-malleable extractors with either a better seed length or a better min-entropy requirement. For simplicity and clarity, we will just assume $t=1$, i.e., there is only one tampering adversary. This is also the most interesting case for standard privacy amplification protocols.

Note that our previous $\nipm$ construction implies  Theorem $\ref{basic_rec_nipm}$, which we restate for convenience.
\begingroup
\def\theTHM{\ref{basic_rec_nipm}}
\begin{THM}[restated]
For all integers $m,L>0$, any $\epsilon>0$,   there exists an explicit $(L,1,0,\epsilon,\epsilon')$-$\nipm:\zo^{mL}\times \zo^{d} \rightarrow \zo^{m'}$, where $d=2^{O(\sqrt{\log L})}\log(m/\epsilon), m'= \frac{m}{2^{\sqrt{\log L}}}-2^{O(\sqrt{\log L})}\log(m/\epsilon)$ and $\epsilon'= O(\epsilon L)$.
\end{THM}
\addtocounter{THM}{-1}
\endgroup

We start by proving the following lemma.

\begin{lemma}\label{gradual_nipm}
For all integers $m,L>0$, any $\epsilon>0$, if there is an explicit $(L,1,0,\epsilon,\epsilon_1)$-$\nipm_1:\zo^{mL}\times \zo^{d_1} \rightarrow \zo^{m_1}$, with $d_1 \leq 2^{r(\log L)^{1/q}}\log(m/\epsilon), m_1= \frac{m}{2^{s(\log L)^{1-1/q}}}-2^{O((\log L)^{1-1/q})}\log(m/\epsilon)$ and $\epsilon_1\leq g(L)\epsilon L$, where $g(L)$ is a monotonic non-decreasing function of $L$, and $r, s, q$ are parameters, with $q \in \mathbb{N}$, then there is an explicit $(L, 1, 0, \epsilon, \eps_2)$-$\nipm_2:\zo^{mL}\times \zo^{d_2} \rightarrow \zo^{m_2}$, with $d_2 \leq2^{2r^{1-1/(q+1)}(\log L)^{\frac{1}{q+1}}}\log(m/\epsilon), m_2 =\frac{m}{2^{sr^{\frac{1}{q+1}}(\log L)^{1-1/(q+1)}}}-2^{O({(\log L)^{1-1/(q+1)}})} \log(m/\epsilon))$ and  $\eps_2 \leq 2\eps_1$.
\end{lemma} 

\begin{proof}
The idea is to use Algorithm $\ref{alg1}$, with $(\ell,1,0,\epsilon,\epsilon_1')$-$\nipm_1$, $\epsilon_1' \le g(\ell)\ell\epsilon$ as the simpler merger for some parameter $\ell$ s.t. in each step, the merger acts on $\ell$ rows. Following the proof of Theorem $\ref{rec_thm}$, it can be shown that the seed length of $\nipm_2$ will be $$d_2=\log(m/\epsilon)2^{r(\log \ell)^{1/q}} 2^{2\frac{\log L}{\log \ell}}.$$ We now choose an $\ell$ to minimize this, which gives $(\log \ell)^{\frac{q+1}{q}}=\frac{2\log L}{r}$,  and  thus the seed length  is $$d_2=2^{2r^{\frac{q}{q+1}}(2\log L)^{\frac{1}{q+1}}}\log(m/\epsilon).$$  It can be verified that for this setting of parameters, the output length is 
\begin{align*}
m_2&=\frac{m}{(2^{s(\log L)^{1-1/q}})^{\frac{\log L}{\log \ell}}}-O(\ell \log(m/\epsilon)) \\
       &=\frac{m}{2^{sr^{\frac{1}{q+1}}(\log L)^{\frac{q}{q+1}}}}-2^{O({(\frac{\log L}{r})^{\frac{q}{q+1}}})} \log(m/\epsilon)\\
       &=\frac{m}{2^{sr^{\frac{1}{q+1}}(\log L)^{1-1/(q+1)}}}-2^{O({(\log L)^{\frac{q}{q+1}}})} \log(m/\epsilon))  
\end{align*}

Finally, the error is bounded by $\sum_{i=1}^{\frac{\log L}{\log \ell}}g(\ell)\epsilon\ell^i <2 g(\ell)L\epsilon<2\epsilon_1$. 
\end{proof}

Now, starting with the $\nipm$ from Theorem $\ref{basic_rec_nipm}$, and using  Lemma $\ref{gradual_nipm}$ an optimal number of times, we have the following theorem.
\begingroup
\def\theTHM{\ref{advanced_nipm}}
\begin{THM}[restated] For all integers $m,L>0$, any $\epsilon>0$,   there exists an explicit $(L,1,0,\epsilon,\epsilon')$-$\nipm:\zo^{mL}\times \zo^{d} \rightarrow \zo^{m'}$, where $d=2^{O(\sqrt{\log \log L})}\log (m/\eps), m'= \frac{m}{L2^{(\log \log L)^{O(1)}}}-O(L\log(m/\eps))$ and $\epsilon'= 2^{O(\sqrt{\log \log L})} L \eps$.
\end{THM}
\addtocounter{THM}{-1}
\endgroup
\begin{proof}
We start from the basic case with the $(L,1,0,\epsilon,\epsilon')-\nipm$ from Theorem $\ref{basic_rec_nipm}$. Thus $q=2, r=O(1), s=1$. We now use Lemma $\ref{gradual_nipm}$, increasing $q$ by one each time. Eventually, we stop at $q=\sqrt{\log \log L}$, noticing that this minimize the seed length. It can be verified that the seed length of the final $\nipm$ is $2^{O(\sqrt{\log \log L})}\log (m/\eps)$, the output length is $\frac{m}{L2^{(\log \log L)^{O(1)}}}-O(L\log(m/\eps))$ and the error is bounded by $\eps \leq 2^{O(\sqrt{\log \log L})} L \eps$.
\end{proof}

Using the $\nipm$ from Theorem $\ref{advanced_nipm}$ in Algorithm $\ref{alg3}$, we obtain  the following non-malleable extractor with a slightly shorter seed length than Theorem $\ref{nmext}$ at the expense of requiring larger min-entropy.
\begingroup
\def\theTHM{\ref{nmext_better_seed}}
\begin{THM}[restated] For all $n, k \in \mathbb{N}$ and any $\epsilon>0$, with  $k \ge (\log(n/\epsilon))^32^{(\log \log \log(n/\epsilon))^{O(1)}}$, there exists an explicit $(k,\epsilon)$-non-malleable extractor $\nmExt:\zo^n \times \zo^d \rightarrow \zo^{m}$, where $d= \log(n/\epsilon)2^{2^{O(\sqrt{\log \log \log(n/\epsilon)})}}, m= \frac{k}{\log(n/\epsilon)2^{(\log \log \log(n/\epsilon))^{O(1)}}} - O((\log(n/\epsilon))^2)$.
\end{THM}
\addtocounter{THM}{-1}
\endgroup
The proof of Theorem $\ref{nmext_better_seed}$ is exactly similar to Theorem $\ref{nmext}$, and we skip it.

It is not hard to modify Algorithm $\ref{alg3}$ such that the  the role of the source and the seed are swapped, in the sense that the seed to $\nipm$ is a deterministic function of the source to the non-malleable extractor, and the matrix is a deterministic function of the seed to the non-malleable extractor.  By this modification. we can achieve a non-malleable extractor that works for lower slightly min-entropy than Theorem $\ref{nmext}$  at the expense of using a larger seed. We state the following theorem without proof.
\begingroup
\def\theTHM{\ref{nmext_better_entropy}}
\begin{THM}[restated] For all $n, k \in \mathbb{N}$ and any  $\epsilon>0$, with  $k \ge \log(n/\epsilon)2^{2^{\Omega(\sqrt{\log \log \log(n/\epsilon)})}}$, there exists an explicit $(k,\epsilon)$-non-malleable extractor $\nmExt:\zo^n \times \zo^d \rightarrow \zo^{m}$, where $d= (\log(n/\epsilon))^32^{(\log \log \log(n/\epsilon))^{O(1)}}, m= \Omega(k)$.
\end{THM}
\addtocounter{THM}{-1}
\endgroup

%% file: optimal_nm.bbl
\newcommand{\etalchar}[1]{$^{#1}$}
\begin{thebibliography}{LRVW03}

\bibitem[ADJ{\etalchar{+}}14]{ADD14}
Divesh Aggarwal, Yevgeniy Dodis, Zahra Jafargholi, Eric Miles, and Leonid
  Reyzin.
\newblock {\em Advances in Cryptology -- CRYPTO 2014: 34th Annual Cryptology
  Conference, Santa Barbara, CA, USA, August 17-21, 2014, Proceedings, Part
  II}, chapter Amplifying Privacy in Privacy Amplification, pages 183--198.
\newblock Springer Berlin Heidelberg, Berlin, Heidelberg, 2014.

\bibitem[AGM03]{AGM03}
Noga Alon, Oded Goldreich, and Yishay Mansour.
\newblock Almost k-wise independence versus k-wise independence.
\newblock {\em Inf. Process. Lett.}, 88(3):107--110, 2003.

\bibitem[AHL15]{AHL15}
Divesh Aggarwal, Kaave Hosseini, and Shachar Lovett.
\newblock Affine-malleable extractors, spectrum doubling, and application to
  privacy amplification.
\newblock Cryptology ePrint Archive, Report 2015/1094, 2015.
\newblock \url{http://eprint.iacr.org/}.

\bibitem[BBCM95]{BBCM95}
C.~H. Bennett, G.~Brassard, C.~Crepeau, and U.~M. Maurer.
\newblock Generalized privacy amplification.
\newblock {\em IEEE Transactions on Information Theory}, 41(6):1915--1923, Nov
  1995.

\bibitem[BBR88]{BBR88}
C.H. Bennett, G.~Brassard, and J.-M. Robert.
\newblock Privacy amplification by public discussion.
\newblock {\em SIAM Journal on Computing}, 17:210--229, 1988.

\bibitem[BIW06]{BIW}
Boaz Barak, Russell Impagliazzo, and Avi Wigderson.
\newblock Extracting randomness using few independent sources.
\newblock {\em SIAM J. Comput.}, 36(4):1095--1118, December 2006.

\bibitem[BKS{\etalchar{+}}10]{BKSSW10}
Boaz Barak, Guy Kindler, Ronen Shaltiel, Benny Sudakov, and Avi Wigderson.
\newblock Simulating independence: New constructions of condensers, {R}amsey
  graphs, dispersers, and extractors.
\newblock {\em J. {ACM}}, 57(4), 2010.

\bibitem[Bou05]{B2}
J.~Bourgain.
\newblock More on the sum-product phenomenon in prime fields and its
  applications.
\newblock {\em International Journal of Number Theory}, 01(01):1--32, 2005.

\bibitem[BRSW12]{BRSW12}
Boaz Barak, Anup Rao, Ronen Shaltiel, and Avi Wigderson.
\newblock 2-source dispersers for $n^{o (1)}$ entropy, and {R}amsey graphs
  beating the {F}rankl-{W}ilson construction.
\newblock {\em Annals of Mathematics}, 176(3):1483--1543, 2012.
\newblock Preliminary version in STOC '06.

\bibitem[CG88]{CG88}
Benny Chor and Oded Goldreich.
\newblock Unbiased bits from sources of weak randomness and probabilistic
  communication complexity.
\newblock {\em SIAM Journal on Computing}, 17(2):230--261, 1988.

\bibitem[CGL16]{CGL15}
Eshan Chattopadhyay, Vipul Goyal, and Xin Li.
\newblock Non-malleable extractors and codes, with their many tampered
  extensions.
\newblock In {\em STOC}, 2016.

\bibitem[CKOR10]{ckor}
N.~Chandran, B.~Kanukurthi, R.~Ostrovsky, and L.~Reyzin.
\newblock Privacy amplification with asymptotically optimal entropy loss.
\newblock In {\em Proceedings of the 42nd Annual ACM Symposium on Theory of
  Computing}, pages 785--794, 2010.

\bibitem[CL16]{CL15}
Eshan Chattopadhyay and Xin Li.
\newblock Extractors for sumset sources.
\newblock In {\em STOC}, 2016.

\bibitem[Coh15]{Coh15}
Gil Cohen.
\newblock Local correlation breakers and applications to three-source
  extractors and mergers.
\newblock In {\em Proceedings of the 56th Annual IEEE Symposium on Foundations
  of Computer Science}, 2015.

\bibitem[Coh16a]{Coh15nm}
Gil Cohen.
\newblock Non-malleable extractors - new tools and improved constructions.
\newblock In {\em CCC}, 2016.

\bibitem[Coh16b]{Coh16a}
Gil Cohen.
\newblock Non-malleable extractors with logarithmic seeds.
\newblock Technical Report TR16-030, ECCC, 2016.

\bibitem[Coh16c]{Coh15b}
Gil Cohen.
\newblock Two-source dispersers for polylogarithmic entropy and improved
  {R}amsey graphs.
\newblock In {\em STOC}, 2016.

\bibitem[CRS14]{CRS12}
Gil Cohen, Ran Raz, and Gil Segev.
\newblock Non-malleable extractors with short seeds and applications to privacy
  amplification.
\newblock {\em SIAM Journal on Computing}, 43(2):450--476, 2014.

\bibitem[CS16]{Coh16b}
Gil Cohen and Leonard Schulman.
\newblock Extractors for near logarithmic min-entropy.
\newblock Technical Report TR16-014, ECCC, 2016.

\bibitem[CZ16]{CZ15}
Eshan Chattopadhyay and David Zuckerman.
\newblock Explicit two-source extractors and resilient functions.
\newblock In {\em STOC}, 2016.

\bibitem[DGJ{\etalchar{+}}10]{DGJSV}
Ilias Diakonikolas, Parikshit Gopalan, Ragesh Jaiswal, Rocco~A. Servedio, and
  Emanuele Viola.
\newblock Bounded independence fools halfspaces.
\newblock {\em SIAM Journal on Computing}, 39(8):3441--3462, 2010.

\bibitem[DKRS06]{dkrs}
Y.~Dodis, J.~Katz, L.~Reyzin, and A.~Smith.
\newblock Robust fuzzy extractors and authenticated key agreement from close
  secrets.
\newblock In {\em Advances in Cryptology --- CRYPTO '06, 26th Annual
  International Cryptology Conference, Proceedings}, pages 232--250, 2006.

\bibitem[DKSS09]{DKSS09}
Zeev Dvir, Swastik Kopparty, Shubhangi Saraf, and Madhu Sudan.
\newblock Extensions to the method of multiplicities, with applications to
  {K}akeya sets and mergers.
\newblock In {\em Proceedings of the 50th Annual IEEE Symposium on Foundations
  of Computer Science}, pages 181--190, 2009.

\bibitem[DLWZ14]{DLWZ11}
Yevgeniy Dodis, Xin Li, Trevor~D. Wooley, and David Zuckerman.
\newblock Privacy amplification and non-malleable extractors via character
  sums.
\newblock {\em SIAM Journal on Computing}, 43(2):800--830, 2014.

\bibitem[DORS08]{DORS08}
Y.~Dodis, R.~Ostrovsky, L.~Reyzin, and A.~Smith.
\newblock Fuzzy extractors: How to generate strong keys from biometrics and
  other noisy data.
\newblock {\em SIAM Journal on Computing}, 38:97--139, 2008.

\bibitem[DP07]{DP07}
Stefan Dziembowski and Krzysztof Pietrzak.
\newblock Intrusion-resilient secret sharing.
\newblock In {\em Proceedings of the 48th Annual IEEE Symposium on Foundations
  of Computer Science}, FOCS '07, pages 227--237, Washington, DC, USA, 2007.
  IEEE Computer Society.

\bibitem[DW09]{DW09}
Yevgeniy Dodis and Daniel Wichs.
\newblock Non-malleable extractors and symmetric key cryptography from weak
  secrets.
\newblock In {\em STOC}, pages 601--610, 2009.

\bibitem[DY13]{DY12}
Yevgeniy Dodis and Yu~Yu.
\newblock Overcoming weak expectations.
\newblock In {\em 10th Theory of Cryptography Conference}, 2013.

\bibitem[GUV09]{GUV09}
Venkatesan Guruswami, Christopher Umans, and Salil~P. Vadhan.
\newblock Unbalanced expanders and randomness extractors from
  {P}arvaresh--{V}ardy codes.
\newblock {\em J. ACM}, 56(4), 2009.

\bibitem[KR09]{KR09}
B.~Kanukurthi and L.~Reyzin.
\newblock Key agreement from close secrets over unsecured channels.
\newblock In {\em EUROCRYPT 2009, 28th Annual International Conference on the
  Theory and Applications of Cryptographic Techniques}, 2009.

\bibitem[Li11]{Li11}
Xin Li.
\newblock Improved constructions of three source extractors.
\newblock In {\em Proceedings of the 26th Annual {IEEE} Conference on
  Computational Complexity, {CCC} 2011, San Jose, California, June 8-10, 2011},
  pages 126--136, 2011.

\bibitem[Li12a]{Li12a}
Xin Li.
\newblock Design extractors, non-malleable condensers and privacy
  amplification.
\newblock In {\em Proceedings of the 44th Annual ACM Symposium on Theory of
  Computing}, pages 837--854, 2012.

\bibitem[Li12b]{Li12b}
Xin Li.
\newblock Non-malleable extractors, two-source extractors and privacy
  amplification.
\newblock In {\em Proceedings of the 53rd Annual IEEE Symposium on Foundations
  of Computer Science}, pages 688--697, 2012.

\bibitem[Li13a]{Li13b}
Xin Li.
\newblock Extractors for a constant number of independent sources with
  polylogarithmic min-entropy.
\newblock In {\em Proceedings of the 54th Annual IEEE Symposium on Foundations
  of Computer Science}, pages 100--109, 2013.

\bibitem[Li13b]{Li13a}
Xin Li.
\newblock New independent source extractors with exponential improvement.
\newblock In {\em Proceedings of the 45th Annual ACM Symposium on Theory of
  Computing}, pages 783--792, 2013.

\bibitem[Li15a]{Li:2source}
Xin Li.
\newblock Improved constructions of two-source extractors.
\newblock Technical Report TR15-125, ECCC, 2015.

\bibitem[Li15b]{Li:affine}
Xin Li.
\newblock Improved two-source extractors, and affine extractors for
  polylogarithmic entropy.
\newblock Technical Report TR15-125, ECCC, 2015.

\bibitem[Li15c]{Li15b}
Xin Li.
\newblock Non-malleable condensers for arbitrary min-entropy, and almost
  optimal protocols for privacy amplification.
\newblock In {\em 12th Theory of Cryptography Conference}, 2015.

\bibitem[Li15d]{Li15c}
Xin Li.
\newblock Three-source extractors for polylogarithmic min-entropy.
\newblock In {\em Proceedings of the 56th Annual IEEE Symposium on Foundations
  of Computer Science}, 2015.

\bibitem[LRVW03]{LRVW03}
Chi-Jen Lu, Omer Reingold, Salil~P. Vadhan, and Avi Wigderson.
\newblock Extractors: optimal up to constant factors.
\newblock In {\em STOC}, pages 602--611, 2003.

\bibitem[Mau92]{Mau92}
Ueli~M. Maurer.
\newblock Conditionally-perfect secrecy and a provably-secure randomized
  cipher.
\newblock {\em Journal of Cryptology}, 5(1):53--66, 1992.

\bibitem[Mek15]{Mek:resil}
Raghu Meka.
\newblock Explicit resilient functions matching {Ajtai-Linial}.
\newblock {\em CoRR}, abs/1509.00092, 2015.

\bibitem[MW97]{MW07}
Ueli Maurer and Stefan Wolf.
\newblock Privacy amplification secure against active adversaries.
\newblock In {\em Advances in Cryptology --- CRYPTO '97}, volume 1294, pages
  307--321, August 1997.

\bibitem[NZ96]{NZ96}
Noam Nisan and David Zuckerman.
\newblock Randomness is linear in space.
\newblock {\em J. Comput. Syst. Sci.}, 52(1):43--52, 1996.

\bibitem[Rao09]{Rao06}
Anup Rao.
\newblock Extractors for a constant number of polynomially small min-entropy
  independent sources.
\newblock {\em {SIAM} J. Comput.}, 39(1):168--194, 2009.

\bibitem[Raz05]{Raz05}
Ran Raz.
\newblock Extractors with weak random seeds.
\newblock In {\em Proceedings of the 37th Annual ACM Symposium on Theory of
  Computing}, pages 11--20, 2005.

\bibitem[RW03]{RW03}
Renato Renner and Stefan Wolf.
\newblock Unconditional authenticity and privacy from an arbitrarily weak
  secret.
\newblock In {\em Advances in Cryptology --- CRYPTO '03, 23rd Annual
  International Cryptology Conference, Proceedings}, pages 78--95, 2003.

\bibitem[RZ08]{RZ08}
Anup Rao and David Zuckerman.
\newblock Extractors for three uneven-length sources.
\newblock In {\em Approximation, Randomization and Combinatorial Optimization.
  Algorithms and Techniques, 11th International Workshop, {APPROX} 2008, and
  12th International Workshop, {RANDOM} 2008, Boston, MA, USA, August 25-27,
  2008. Proceedings}, pages 557--570, 2008.

\bibitem[Vio14]{Viola14}
Emanuele Viola.
\newblock Extractors for circuit sources.
\newblock {\em {SIAM} J. Comput.}, 43(2):655--672, 2014.

\end{thebibliography}
